\documentclass{article}

\usepackage{blindtext} 

\usepackage{amsmath,amssymb,amsfonts,amsthm,enumerate}
\usepackage[sc]{mathpazo} 
\usepackage[T1]{fontenc} 
\usepackage{microtype} 
\usepackage{fullpage}
\usepackage{graphicx}
\usepackage[english]{babel} 

\usepackage[hmarginratio=1:1,top=32mm,columnsep=20pt]{geometry} 
\usepackage[hang, small,labelfont=bf,up,textfont=it,up]{caption} 
\usepackage{booktabs} 

\usepackage{lettrine} 

\usepackage{enumitem} 
\setlist[itemize]{noitemsep} 

\usepackage{abstract} 

\usepackage{titlesec} 
\renewcommand\thesection{\Roman{section}} 
\renewcommand\thesubsection{\roman{subsection}} 
\titleformat{\section}[block]{\large\scshape\centering}{\thesection.}{1em}{} 
\titleformat{\subsection}[block]{\large}{\thesubsection.}{1em}{} 

\usepackage{fancyhdr} 

\usepackage{titling} 

\usepackage{hyperref} 
\newtheorem{Theor} {Theorem}
\newtheorem{rk}{Remark}
 
\newtheorem{defi}{Definition}
\newtheorem{prop}{Proposition}
\newtheorem{lem}{Lemma}
\newtheorem{conj}{Conjecture}

\pretitle{\begin{center}\Huge\bfseries} 
\posttitle{\end{center}} 
\title{RIP constants for deterministic compressed sensing matrices-beyond Gershgorin} 
\author{%
\textsc{Arman Arian and \"{O}zg\"{u}r Y\i lmaz} \\[1ex] 
\normalsize Department of Mathematics, \\ University of British Columbia \\ 
}
\date{} 
\renewcommand{\maketitlehookd}{%
\begin{abstract}

Compressed sensing (CS) is a signal acquisition paradigm to simultaneously acquire and reduce dimension of signals that admit sparse representations. This is achieved by collecting linear, non-adaptive measurements of a signal, which can be formalized as multiplying the signal with a ``measurement matrix". If the measurement matrix satisfies the so-called restricted isometry property (RIP), then it will be appropriate to be using in compressed sensing. While a wide class of random matrices provably satisfy the RIP with high probability, explicit and deterministic constructions have been shown (so far) to satisfy the RIP only in a significantly suboptimal regime. 

In this paper, we propose two novel approaches for improving the RIP constant estimates based on Gershgorin circle theorem for a specific deterministic construction based on Paley tight frames, obtaining an improvement over the Gershgorin bounds by a multiplicative constant. In one approach we use a recent result on the spectra of the skew-adjacency matrices of oriented graphs. In the other approach we use the so-called Dembo bounds on the extreme eigenvalues of a positive semidefinite Hermitian matrix. We also generalize these bounds and we combine the new bounds with a conjecture we make regarding the distribution of quadratic residues in a finite field to provide a potential path to break the so-called ``square-root barrier"--we give a proof based on the assumption that the conjecture holds.

Index Terms--compressed sensing, deterministic matrices, Paley tight frame, Gershgorin bounds, Dembo bounds, square-root barrier, quadratic residues.

\end{abstract}
}

\begin{document}

\maketitle


\label{breaksquare}

Given a deterministic CS matrix, one of the most common ways to bound its RIP constant is by relating its RIP constant with its \textit{mutual coherence} via  

\begin{equation} \label{maxsparsity1} \delta_k \leq \mu (k-1) \end{equation}


\noindent Throughout this paper, we will call this bound, the Gershgorin bound on the RIP constants, or simply the Gershgorin bound as this is the bound obtained from Gershgorin circle theorem.  Moreover, the tightest bound that relates the RIP constants of a matrix to its performance as a CS measurement matrix is due to \cite{cai2} and it states that to ensure recovery of $k$-sparse (or compressible) signals, we need $\delta_{2k} < \frac{1}{\sqrt{2}}$. On the other hand, using \eqref{maxsparsity1}, in order for a matrix to have small enough RIP constant, it is sufficient that the maximum sparsity level satisfies $k < \frac{1}{2 \mu \sqrt{2}} +\frac{1}{2}$.  Considering the Welch bound for the coherence of an $m \times n$ matrix,  this imposes a square-root barrier on the sparsity level of signals, namely, $k=\mathcal{O}(\sqrt{m})$. Comparing this level of sparsity with the maximum level of sparsity for sub-Gaussian matrices, which is found  to be $k=\mathcal{O}(\frac{m}{ \log \frac{n}{m}})$, we observe that there is a huge difference between these two. 

 In fact, as mentioned in \cite{foucart2}, p. 141, finding a deterministic CS matrix that satisfies RIP in the optimal regime is a major open problem. Here we quote a few sentences from \cite{foucart2} that explains the intrinsic difficulty of reaching RIP in the optimal regime : 

``The intrinsic difficulty in bounding the restricted isometry constants of explicit matrices $A$ lies in the basic tool for estimating the eigenvalues of $A_{S}^*A_{S}  -Id$, namely, Gershgorin's disk Theorem ... the quadratic bottleneck is unavoidable when using Gershgorin's theorem to estimate restricted isometry constants. It seems that not only the magnitude of the entries of Gramian $A^*A$ but also their signs should be taken into account in order to improve estimates for deterministic matrices, but which tools to be used for this purpose remain unclear." We will verify the fact that one needs to take account the \textit{signs} of entries of the Gramian matrix (in addition to the magnitudes) to obtain a bound that improves the Gershgorin bound. See Section \ref{paleytight}. Moreover, it is shown in \cite{wangplan} that if the number of measurements satisfies $m \gg k^{1+\alpha}s^{-2} \log n$, with $\alpha \in [0,1)$, there is no polynomial time algorithm that can \textit{certify} the RIP constants satisfy $\delta_k \leq s$. This shows the significance of finding RIP matrices, even in the suboptimal regime $m \gg k^{1+\alpha}s^{-2} \log n$, with $\alpha \in [0,1)$. The first step to do so, is going \textit{beyond} the Gershgorin bound, namely, the bound given in \eqref{maxsparsity1}.

In this chapter, we will propose two different tools to replace Gershgorin circle theorem for bounding eigenvalues of the Gramian matrix to estimate the isometry constants for a specific construction. In the first approach, we compare the Gramian matrices of this construction with the skew-adjacency matrices of specific graphs to obtain bounds on the extreme eigenvalues of the Gramian matrices and hence, will estimate the RIP constants. 

To explain the idea used in the second approach, first note that Gershgorin circle theorem bounds \textit{every} eigenvalue of a matrix uniformly. That is, it does not distinguish the extreme eigenvalues with other eigenvalues and it states that every eigenvalue lies in one of Gershgorin circles. However, the isometry constants only depend on the minimum and maximum eigenvalues of the Gramian matrix. There is in fact, a bound called ``Dembo bound" which provides bounds for the maximum and minimum eigenvalues of a positive semidefinite Hermitian matrix. The goal in this chapter is to use one of these two approaches  to achieve an improved bound for the isometry constant, i.e., something better than $\delta_k \leq (k-1) \mu$. We will see that using the first approach mentioned above, one can improve the classical Gerhsgorin bound by a multiplicative constant while using the second approach one can have a small additive improvement. However, the second approach has its own significance because using this approach we will give a pathway by providing an explicit conjecture regarding the distribution of quadratic residues, to break the square-root barrier via $k=\mathcal{O}(m^{5/7})$ (if the conjecture holds). 

All results in the literature on sparse recovery using the standard RIP rely on the Welch bound or its variants using $\ell_1$-coherence. The only exception to this, until our work, is the work of Bourgain et al. \cite{bourgain2}. For a prime number $p$, they constructed an explicit CS matrix of the order $p \times p^{1+\epsilon}$ (where $\epsilon>0$ is a small number and $m=p$ is the number of measurements) such that this matrix satisfies RIP with $\delta_k < 1/ \sqrt{2}$ when $k=\lfloor p^{\frac{1}{2}+\epsilon_0} \rfloor$  (with $\epsilon_0<\epsilon$ is also a small constant) and $p$ is large enough.  As mentioned above, while we can not break the square-root barrier, we will propose novel approaches to improve the bounds based on coherence or $\ell_1$-coherence. Lastly, we will propose a conjecture that if it holds, we would have an \textit{improved} version of breaking the square-root barrier compared to the one given in \cite{bourgain2}. This improvement will be on how close to unity the power $\alpha$ can be chosen in $k=\mathcal{O}(m^{\alpha})$, and on the lower bound on the minimum number of measurements.

\section{Paley tight frames for compressed sensing} \label{paleytight}

In this chapter, we will investigate the behaviour of the RIP constants of a specific class of matrices, and will show that it behaves better than what is expected using the Gerhsgorin circle theorem, i.e., the bound given by \eqref{maxsparsity1}. In order to choose such a class of matrices, first note that for a (normalized) measurement matrix $\Phi$, with the coherence $\mu$, the $2 \times 2$ Gramian matrix of the form $\begin{bmatrix} 1 & c \\ c^* & 1 \end{bmatrix} $, with $|c|=\mu$, has the extreme eigenvalues $1\pm \mu$, and hence, $\delta_2=\mu$ as predicted by \eqref{maxsparsity1}. In the next step, we consider a Gramian matrix of order 3 of the form $\begin{bmatrix} 1 & \mu & \mu \\ \mu & 1 & \mu \\ \mu & \mu & 1 \end{bmatrix}$, and we observe that the extreme eigenvalues of this matrix are of the form $1\pm 2\mu$. However, if we consider a Gramian matrix of the form $\begin{bmatrix} 1 & i \mu & i \mu \\ -i \mu & 1 & \mu \\ -i \mu & -i \mu & 1 \end{bmatrix}$ (with $i=\sqrt{-1}$), the extreme eigenvalues are of the form $1 \pm \sqrt{3} \mu$. Moreover, it can be seen that for a larger value of $k$, the spectral radius of the Gramian matrix of order $k$ can reduce further if all non-diagonal entries are imaginary numbers and a mixture of the above diagonal entries have negative imaginary parts (as opposed to all above diagonal entries having positive imaginary parts, or all having negative imaginary parts).  Therefore, we search among measurement matrices with the property that the inner product of distinct columns are imaginary numbers, and also for large enough $k$, a mixture of above-diagonal entries of Gramian matrices of order $k$, have negative imaginary parts. Such a construction is based on  \textit{Paley tight frame} as proposed in \cite{fickus}. Specifically, we will consider the following matrices.  

\begin{defi} \label{tilde} 
Let $p \equiv 3 $ mod 4 be a prime number, and consider the $p \times p $ DFT matrix whose $(m,n)$th entry is given by $e^{\frac{2 \pi i}{p} mn}$. Next, choose the $(p+1)/2$ rows of the DFT matrix  whose indices  are quadratic residues mod $p$ (starting with the row corresponding to $m=0$). We denote this $(p+1)/2 \times p$ matrix by $H$, which we normalize to obtain the measurement matrix $\Phi$ : $$\Phi :=DH,$$ where $D$ is the diagonal matrix whose first diagonal entry is $\sqrt{\frac{1}{p}}$, and the rest of its diagonal entries are $\sqrt{\frac{2}{p}}$.

\end{defi} 

  Hence, our measurement matrix is a $(p+1)/2 \times p$ matrix with unit norm columns. For example, for $p=7$, we should consider the $7 \times 7$ DFT matrix, and then consider the 1st, 2nd, 3rd, and 5th rows (corresponding to the quadratic residues $m=0,  1, 2, 4$ in $\mathbb{Z}_7$), and subsequently normalize the resultant matrix as mentioned above to obtain the $4 \times 7$ Paley CS matrix. We cam compute  the inner product between the columns corresponding to $n, n' \in \mathbb{Z}_p$  as follows.  $$\langle \phi_n , \phi_{n'} \rangle =\frac{1}{p}+\frac{2}{p} \sum_{m=1}^{\frac{p-1}{2}} e^{\frac{2 \pi i(n-n')}{p} m^2} = \frac{1}{p} \cdot \sum_{m=0}^{p-1} e^{\frac{2 \pi i (n-n')}{p}m^2} = \Big( \frac{n-n'}{p} \Big) \frac{i}{\sqrt{p}}. $$ 

\noindent Here, $\Big( \frac{n-n'}{p} \Big)$ denotes the Legendre symbol (and is 1 if $(n-n')$ is a quadratic residue, and is -1 if $(n-n')$ is a quadratic non-residue).  

One way to bound the RIP constants of this construction is using \eqref{maxsparsity1}, which gives \begin{equation} \label{againg} \delta_k \leq \frac{k-1}{\sqrt{p}} \end{equation} On the contrary, we observe numerically that at least the lower bound on the RIP constant behaves much better. In fact, what we observe in Figure \ref{aboutj} would be consistent with  \begin{equation} \label{kbeta} \delta_k \leq \frac{k^{\beta}}{\sqrt{p}} \end{equation} for $\beta \approx 0.65 $. Note that if \eqref{kbeta} is proved, then the square-root barrier would be broken. In this paper, we will show that for this construction, the bound \eqref{againg} can be improved by an additive or a multiplicative constant using two novel approaches. We will also propose a conjecture regarding distribution of quadratic residues in $\mathbb{Z}_p$ that leads to \eqref{kbeta}.  Next,  we explain how we obtain Figure \ref{aboutj}.  

Fix a value of $p$, say $p=103$, and consider the Paley CS matrix $\Phi$ as defined above. Also, fix a value of $k$, say $k=30$, and choose a signal with random support $T \subseteq \{1,2,...,p\}$ and with $|T|=k$. Let $T=\{r_1,...,r_k\}$, and for each $1 \leq j \leq k$, let $d(j)$ be an estimation for the RIP constant of order $j$ defined by $$d(j)=\max\{ \lambda_{\max} (G_j)-1,1-\lambda_{\min}(G_j) \}$$ where $G_j :=\Phi_{T_j}^* \Phi_{T_j}$, and $T_j :=\{r_1,...,r_j \}$. The graph of $d(j)$ as a function of $j$ is shown in Figure \ref{aboutj} in log-log scale. In this figure, we also plot the classical Gershgorin bound as well as the new improved bound (which will be derived in Section \ref{specificmatrix}) on the RIP constants. As we observe in this figure and as suggested by the method of least squares, the lower bound function $d(j)$ for the RIP constants behaves like $j^{\beta}$ for $\beta \approx 0.65308$. Note that as mentioned above, $d(j)$ is a lower bound for the RIP constants since it is obtained by using only one random support set, while the precise value of RIP constants are obtained by considering the worst case over exponentially many support sets. Accordingly, we perform another experiment in which we compare the behaviour of $d(j)$ as defined above obtained from a single random support set with $d'(j)$ obtained from the worst case of 1000 random support sets. As we observe in Figure \ref{compared1d2}, as we increase the number of random support sets, the behaviour of RIP constant estimation remains almost the same. We will use the estimated value of $\beta$ in $d(j)=j^{\beta}$ later in this chapter. 

\begin{figure}
\begin{center}
    \includegraphics[width=0.4\textwidth ]{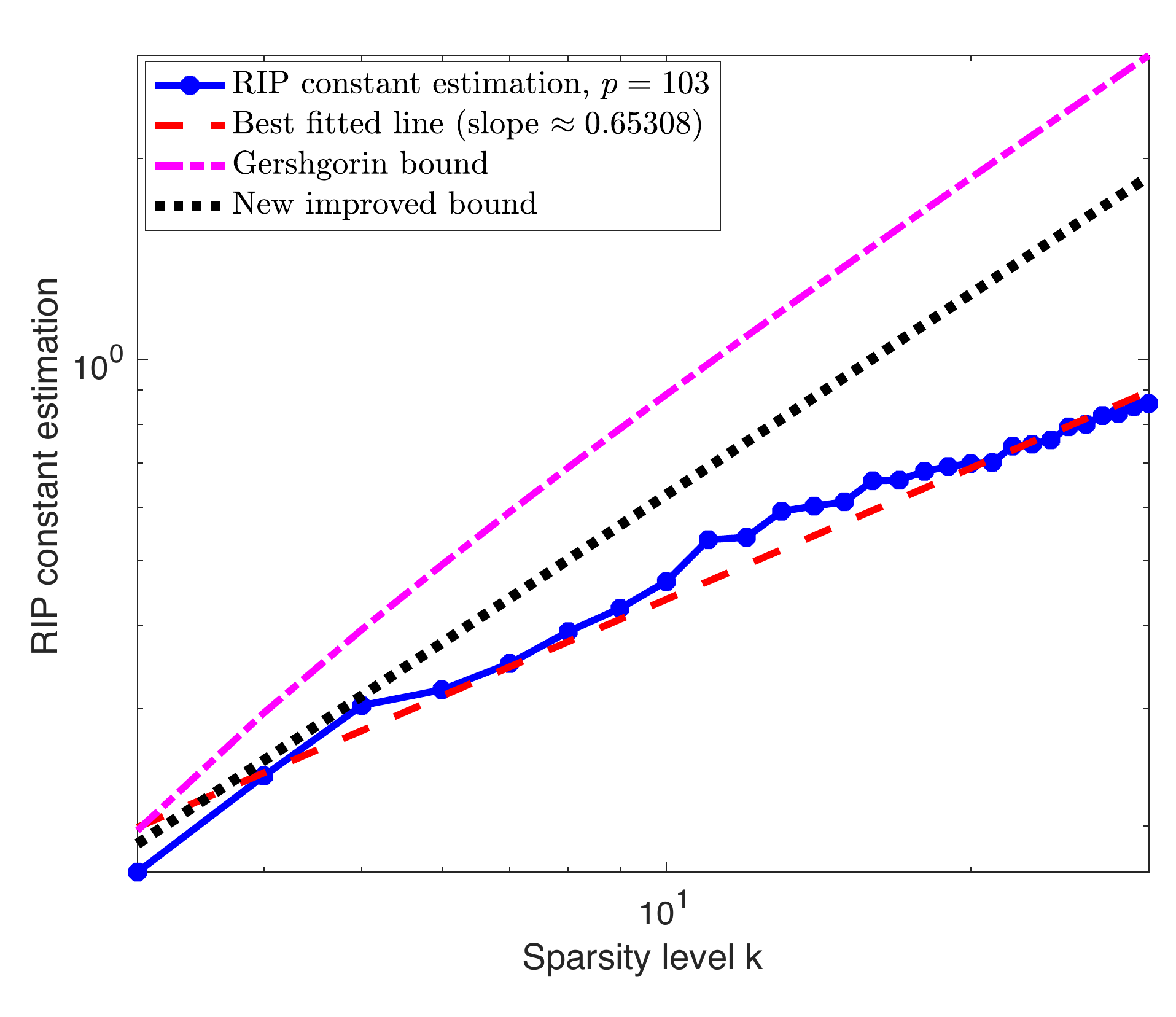}

\caption{ \label{aboutj}   The graph of lower bound of the RIP constants, compared with the Gershgorin bound and the new improved bound, as given in Section \ref{specificmatrix}, on the RIP constants.  }
\end{center}
\end{figure}

\begin{figure}
\begin{center}
    \includegraphics[width=0.4\textwidth ]{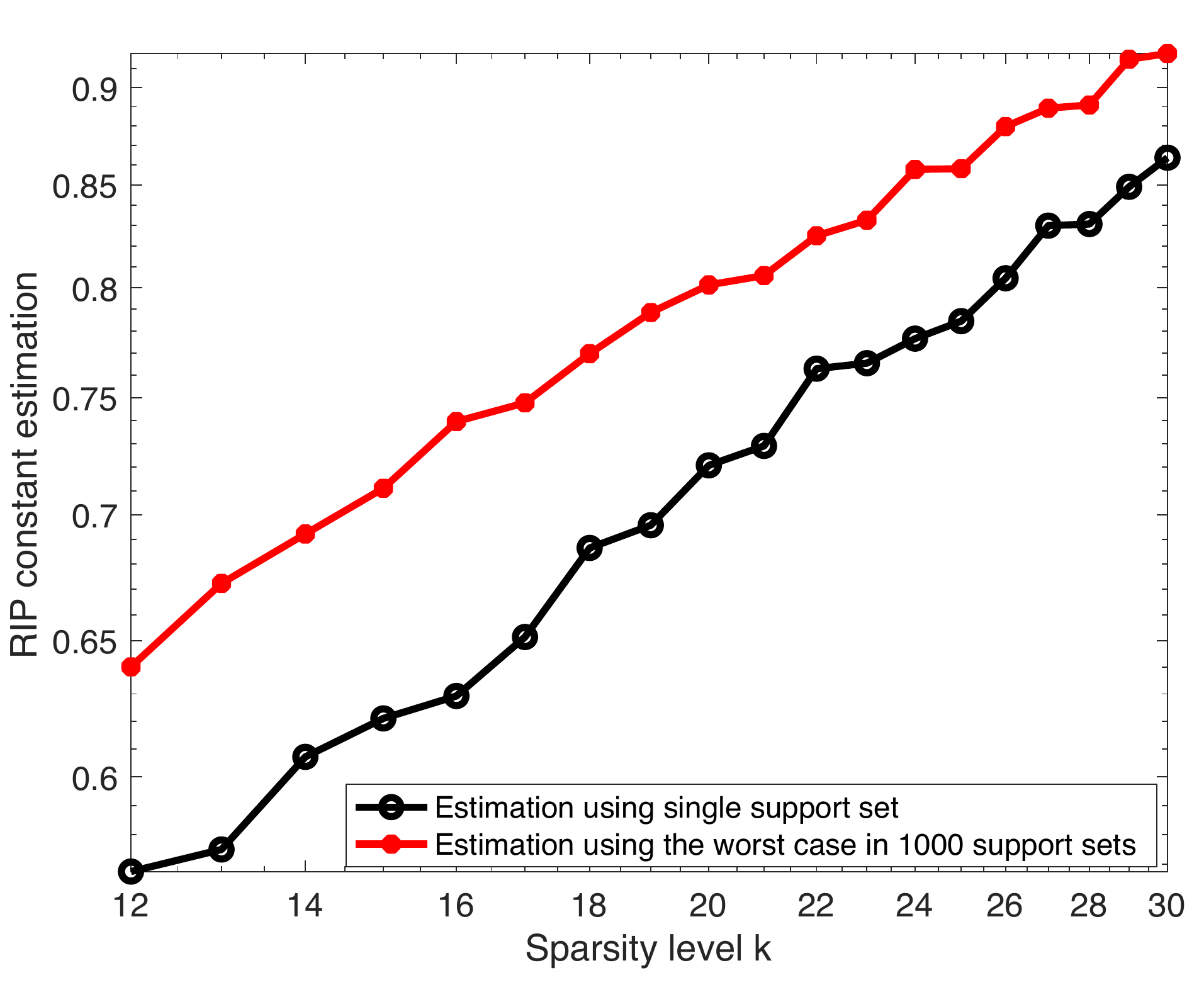}

\caption{ \label{compared1d2}   Comparison of lower bounds of the RIP constants obtained from a single random support set and using the worst case among 1000 random support sets. We observe that the slope of the graph (in log-log scale) obtained from a single support set almost remains constant, as we increase the number of support sets from 1 to 1000. }
\end{center}
\end{figure}


\begin{rk} In the construction used in \cite{fickus}, it is assumed that $p \equiv 1$. However, as our goal in this chapter is to improve the Greshgorin bound, we would not be able to do so with the same assumption. The reason is that the computation above shows that the inner product of $n$th and $n'$th columns of $\Phi$ is given by $\Big( \frac{n-n'}{p} \Big) \cdot \frac{1}{\sqrt{p}}$. Now, consider a set $T=\{r_1,...,r_k\}$ such that $\Big( \frac{r_i-r_j}{p} \Big)=1$ whenever $i<j$. Then, the Gramian matrix will be a $k \times k$ matrix of the following form $$G=\begin{bmatrix} 1 & \frac{1}{\sqrt{p}} & \cdots & \frac{1}{\sqrt{p}}  \\ \vdots \\ \frac{1}{\sqrt{p}} & \frac{1}{\sqrt{p}} & \cdots & 1 \end{bmatrix} $$ The Gershgorin bound for the maximum eigenvalue of this matrix is $\eta=1+\frac{k-1}{\sqrt{p}}$. This is in fact the maximum eigenvalue of this matrix since $$\det \begin{bmatrix} -\frac{k-1}{\sqrt{p}} & \frac{1}{\sqrt{p}} & \cdots & \frac{1}{\sqrt{p}} \\ \vdots \\ \frac{1}{\sqrt{p}} & \frac{1}{\sqrt{p}} & \cdots & -\frac{k-1}{\sqrt{p}} \end{bmatrix} =0 $$ This can be verified by adding rows 2, 3, ..., $k$ to the first row, which makes the first row the zero vector. For this reason, we change the assumption to $p \equiv 3$ mod 4, which as we will see later will lead to improving the Gershgorin bounds. 

\end{rk}

 \section{Improving the Gershgorin bound using skew-adjacency matrices }
 \label{specificmatrix}
 
In this section, we propose an approach that will enable us to improve the Gershgorin bound by a multiplicative constant for the construction given in Definition \ref{tilde}.

%
 
We start by considering the construction given in Definition \ref{tilde}, and decompose the Gramian matrix of order $k$ for this construction, denoted by $G_k$ as follows.  
 
 \begin{equation} \label{theform} G_k=(g_{ij})=I_k+A_k\end{equation}  where $I_k$ is the identity matrix of order $k$, $A_k=(a_{ij})$, $a_{ii}=0$, and  $a_{ij}=\frac{\sqrt{-1}}{\sqrt{p}}$ or $\frac{-\sqrt{-1}}{\sqrt{p}}$ for $i \ne j$.

 Recall that to compute the RIP constant of order $k$ of the measurement matrix $\tilde{\Phi}$, we need to consider the Gramian matrices $G_k^{\max}$, and $G_k^{\min}$  with largest  maximum and smallest minimum  eigenvalues respectively (among all Gramian matrices of the same order). Decompose these matrices as  $G_k^{\max}=I+A_k^{\max}$, and $G_k^{\min}=I+A_k^{\min}$ as in \eqref{theform}. For any matrix $M$, let $\lambda_{\max}(M)$, $\lambda_{\min}(M)$, and $\rho(M)$ denote the maximum and minimum eigenvalues of $M$, and the spectral radius of $M$ respectively.  Then, we have $$\lambda_{\max}(G_k^{\max}) = 1+ \lambda_{\max}(A_k^{\max})$$ 
 
 \noindent and $$\lambda_{\min}(G_k^{\min}) = 1+\lambda_{\min}(A_k^{\min}) $$
 
\noindent On the other hand, each $A_k$ is a skew-symmetric matrix and hence, if $\lambda$ is an eigenvalue of $A_k$, then $-\lambda$ is also an eigenvalue for $A_k$. This means that $\lambda_{\max}(A_k^{\max})=\rho(A_k^{\max})$ and  $\lambda_{\min}(A_k^{\min})=-\rho(A_k^{\max})$. Therefore, \begin{equation} \label{spectraineq1} \delta_k= \max \{\lambda_{\max}(G_k^{\max})-1,1-\lambda_{\min}(G_k^{\min}) \} = \rho(A_k^{\max})\end{equation}

  \noindent It remains to find a bound for $\rho(A_k^{\max})$. Note that $A_k^{\max}$ can be written in the form of \begin{equation} \label{ac} A_k^{\max}=\frac{i}{\sqrt{p}} C_k, \end{equation} where $i=\sqrt{-1}$, and $C_k^{\max}$ is a skew-symmetric matrix with zero diagonals, and whose every other entry is  1, or -1, and it has the largest spectral radius (among all matrices of the same form).  In order to bound the spectral radius of $C_k^{\max}$, we view $C_k^{\max}$ as the skew adjacency of an oriented graph, and use the results in the literature about the spectral radius of these matrices to find bounds on the extreme eigenvalues of $C_k^{\max}$.  First, we need  the following definition.

 
 
 \begin{defi} Let $G$ be a simple undirected graph of order $n$. By $G^{\sigma}$ we denote a directed (or oriented) graph that assigns a direction to every edge of $G$. The skew-symmetric adjacency matrix of $G^{\sigma}$, denoted by $S(G^{\sigma})=(s_{ij})$ is an $n \times n$ skew symmetric matrix such that $s_{i,j}=1$ and $s_{j,i}=-1$ if $i \to j$ is an arc of $G^{\sigma}$. If there is no arc between the vertices $i$ and $j$, we define $s_{i,j}=s_{j,i}=0$. The skew spectral radius of $G^{\sigma}$, denoted by $\rho_S (G^{\sigma})$ is defined as spectral radius of $S(G^{\sigma})$. 
 \end{defi}

 Now, to find a bound for $\rho(A_k^{\max})$, we need to consider the skew adjacency of a simple graph with largest (among all oriented graphs of order $k$) spectral radius. It turns out \cite{deng} that the oriented graph whose skew adjacency matrix has zero diagonals, whose upper diagonals are all 1, and whose lower diagonals are all -1 has the largest spectral radius. In particular, let $K_n$ be the complete graph of order $n$, and let $K_n^{\tau}$ be the oriented complete graph with the adjacency matrix with zero diagonals, with 1's located in the upper diagonal entries, and -1's located in the lower diagonal entries. In other words, $$S(K_n^{\tau})=\begin{bmatrix} 0 & 1 & 1 & \cdots & 1 \\ -1 & 0 & 1 & \cdots & 1 \\ \vdots \\ -1 & -1 & -1 & \cdots & 0 \end{bmatrix} $$ 
 
 \begin{Theor} \label{completecover} For any oriented graph $G^{\sigma}$ of order $n$, $$\rho_S(G^{\sigma}) \leq \rho_S(K_n^{\tau}) =\cot (\frac{\pi}{2n} ) $$ 
 Equality holds if and only if $S(G^{\sigma})=Q^T S(K_n^{\tau}) Q$  for some signed permutation matrix $Q$. 
 \end{Theor}
 
 \noindent Based on this theorem and using our notation, we can conclude that \begin{equation} \label{spectraineq2} \rho(C_k^{max}) \leq \cot (\frac{ \pi}{2k}) \leq \frac{2k}{\pi} \end{equation} 
 
 \noindent where we used the fact that $\cot (x) \leq 1/x $ for $x>0$. This inequality comes from the standard inequality $x<\tan x$ for $x>0$.  Lastly, by combining   \eqref{spectraineq1}, \eqref{ac}, and \eqref{spectraineq2} we obtain  the following theorem. 
 
 \begin{Theor} \label{improvedskew} Let $ p \geq 3$ be a  prime number. The RIP constant of the matrix $\tilde{\Phi}$ as defined in Definition \ref{tilde} satisfies $$\delta_k \leq \frac{2}{\pi} \cdot \frac{k}{\sqrt{p}} $$ for any $k \leq p$.

  \end{Theor}
 
 Therefore, the maximum sparsity level $k$ for which we have a guarantee for  recovery through BPDN using this construction must now satisfy $$2k < \sqrt{\frac{p}{2}} \cdot \frac{\pi}{2}$$ (instead of the standard bound $2k < \sqrt{\frac{p}{2}}+1$). For example, if we set $p=1009$, the maximum sparsity level becomes $2k \leq  \lfloor \sqrt{\frac{p}{2}} \cdot \frac{\pi}{2} \rfloor =35$ (instead of the standard bound $2k \leq \lfloor \sqrt{\frac{p}{2}} \rfloor+1=23$).

 \section{Improving the Gershgorin bound using Dembo bounds}
 \label{improvewithdembo}
   
 We have observed so far that the Gershgorin bound can be improved with a multiplicative constant using the specific construction given in Definition \ref{tilde}. For the rest of this paper, we propose another tool that can also improve the Gershgorin bound. Although this approach improves this bound only by an additive constant but it has advantage over the previous approach in the sense that it can be applied to other constructions. More importantly, we will propose a conjecture that if it holds, then using this approach the square-root barrier can be broken for the construction given in Definition \ref{tilde}. First, we explain the idea behind this approach. 

We know that the RIP constants of a (normalized) measurement matrix $\Phi$ is computed using the extreme eigenvalues of a matrix of the form $A:=\Phi_T^* \Phi_T$, where $T$ is a set $T$ with $|T|=k$. Now, if the coherence of $\Phi$ is $C/\sqrt{m}$ (which is the suboptimal value considering the Welch bound), then by Gershgorin circle theorem, every eigenvalue (including the extreme eigenvalues) lies in $(1-kC/\sqrt{m},1+kC/\sqrt{m})$. However, we would naturally expect to have sharper bounds if we try to bound only the extreme eigenvalues. Below, we attempt to do this using the so-called Dembo bounds, and then we will generalize it later. These bounds estimate the extreme eigenvalues of a positive semidefinite Hermitian matrix, and is the main tool that we use in this approach.

\begin{Theor}  \label{dembobounds} (\textit{Dembo Bounds, \cite{dembo}}) Suppose that a positive semidefinte Hermitian matrix $R$ can be written as $ R=\begin{bmatrix} c & \textbf{b}^* \\ \textbf{b} & Q \end{bmatrix}$ where $Q$ is a $k \times k$ positive semidefinite Hermitian matrix, $\textbf{b}$ is a $k \times 1$ vector, and $c \geq 0$. Also, suppose that $ \lambda_1 \leq \lambda_2 \leq ... \leq \lambda_{k+1}$ are the eigenvalues of $R$. Then, Dembo bounds can be stated as  \begin{equation} \label{biggerlambda} \lambda_{k+1} \leq \frac{c+\eta_k}{2}+\sqrt{\frac{(c-\eta_k)^2}{4}+b^{*}b} \end{equation} and \begin{equation} \label{smallerlambda} \lambda_1 \geq \frac{c+\eta_1}{2}-\sqrt{\frac{(c-\eta_1)^2}{4}+b^{*}b }  \end{equation} where $\eta_1$ is any lower bound on the minimum eigenvalue of $Q$ and $\eta_k$ is any upper bound on the maximum eigenvalue of $Q$.

\end{Theor}

\noindent Before stating and deriving our results formally, we  perform a  numerical experiment.  Consider a Paley CS matrix as defined in Definition \ref{tilde}, with $p=103$. We also fix a $k$-value, say $k=30$, and we choose a random set $T =\{r_1,r_2,...,r_k\}\subseteq \{1,2,...,p\}$. For $1\leq j \leq k$, let $T_j= \{r_1,...,r_j\}$,  $A_j=\Phi_{T_j}$, and $D_j=\Phi_{T_j}^* \Phi_{T_j}$. For $2 \leq j \leq k$, we can write $D_j$ of the form $D_j=\begin{bmatrix} 1 & \textbf{b}^* \\ \textbf{b} & D_{j-1} \end{bmatrix}$.  Let $\lambda_{j-1}$ be the maximum eigenvalue of $D_{j-1}$, and  $\lambda_j$ be the maximum eigenvalue of $D_j$. Since each entry of $\textbf{b}$ is $\pm i/\sqrt{p}$, we have  $\textbf{b}^* \textbf{b} =\frac{j-1}{p}$ . Therefore, using  Dembo bounds, the bound on $\lambda_j$ which we denote by $\eta_j$ can be written as follows.  $$\eta_{j}=\frac{1+\lambda_{j-1}}{2}+\sqrt{\frac{(1-\lambda_{j-1})^2}{4}+\frac{j-1}{p}} $$ We also find the upper bound for the maximum eigenvalue of $D_j$ given by Gershgorin bound, i.e., $$\eta'_j=1+\frac{j-1}{\sqrt{p}}$$ Next,  we calculate the ratio of the upper bounds for the maximum eigenvalues of $D_j$ to the actual maximum eigenvalues of $D_j$ for these two different bounds, i.e.,$\frac{\mbox{Dembo bounds}}{\mbox{actual eigenvalues}}$ and also $\frac{\mbox{Gershgorin bounds}}{\mbox{actual eigenvalues}}$,  and we plot the graphs in log-log scale.  Note that if the graph is closer to $y=1$ it means we have a better and tighter estimate. The graphs are shown in Figure \ref{dembocompare} and we can clearly see that Dembo bound gives a better estimate for the maximum eigenvalues compared to Gershgorin bounds. A similar reasoning can be given regarding the estimates for the minimum eigenvalues. However,  this approach can not be applied in practice since it  assumes that we have access to exact eigenvalues of previous order in each step. Despite this fact,  using Dembo bounds or generalizations of this  bound can lead to an improvement of Gershgorin bounds. We begin by finding a non-trivial bound for $\delta_k$ for a fixed small value of $k$, i.e., a bound tighter than the one given by Gershgorin bound. Then, we apply Dembo bounds or the so-called ``Generalized Dembo bounds" to obtain non-trivial bounds on $\delta_k$ for the next values of $k$ inductively.


 
\begin{figure}
\begin{center}
    \includegraphics[width=0.4\textwidth ]{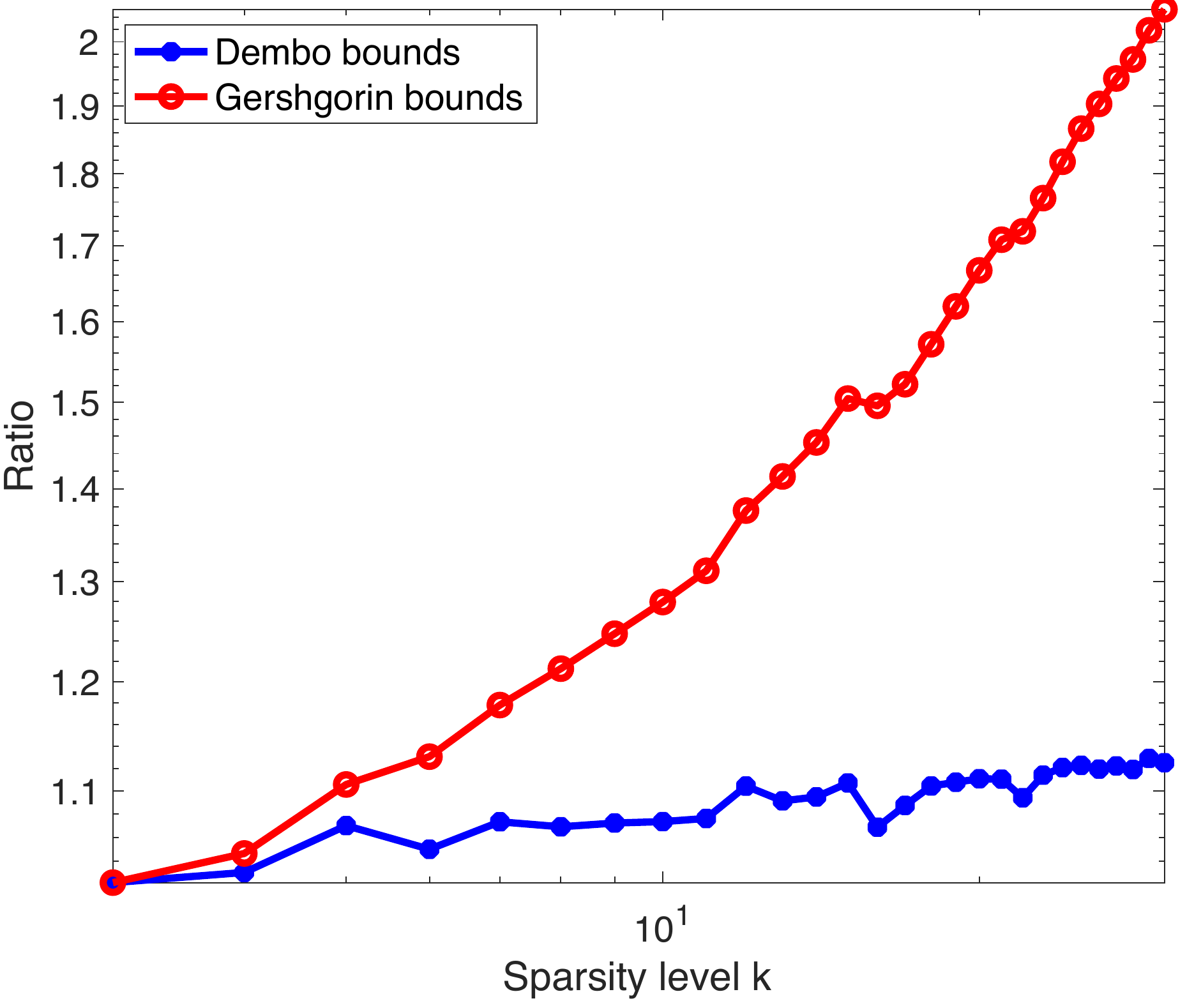}

\caption{ \label{dembocompare}    Comparing the sharpness of Dembo bounds and Gershgorin bounds for estimating the maximum eigenvalue of a semi-positive definite Hermitian matrix by considering the ratio of these bounds over the actual maximum eigenvalues for a fixed Paley matrix with $p=103$.   }
\end{center}
\end{figure}

Theorem below provides a bound for the RIP constants when the construction given in Definition \ref{tilde} are used as the measurement matrices.

\begin{Theor}
\label{mainth}

Let $p \geq 7$ be a prime such that $p \equiv 3 \mod 4$.    Then for $k \geq 3$, the RIP constant $\delta_k$ of the  matrix $\tilde{\Phi}$,  as defined above,  satisfies \begin{equation} \label{improvedbound1} \delta_{k} \leq \frac{k-1-\frac{1}{c(k-1)}}{\sqrt{p}} \end{equation}  where $c=\frac{1}{2(2-\sqrt{3})} $.  \end{Theor}



%
%

\begin{rk}
\label{l1coherence}

The bound given in \eqref{improvedbound1} is not achievable either using Gershgorin bound, or using $\ell_1$-coherence as introduced in \cite{foucart}. The $\ell_1$-coherence function $\mu_1$ of the $m \times n$  matrix $\Phi=[\phi_1  \ \phi_2 \ \cdots \  \phi_n]$ is defined for $s \in [n-1]$ by $$\mu_1 (s):=\max_{i \in [n]} \ \max \{ \sum_{j \in S} | \langle \phi_i, \phi_j \rangle |, S \subseteq [n] , \mbox{card} (S) =s , i \not\in S \}$$ where $[n]=\{1,2,...,n\}$.  It is shown in \cite{foucart} that $\mu \leq \mu_1(s) \leq s \mu $. It is also shown that the RIP constant of a matrix $\Phi$ satisfies $$\delta_k \leq \mu_1 (k-1) \leq (k-1) \mu $$

\noindent So in general $\ell_1$-coherence may give a better bound compared to $\delta_k \leq (k-1) \mu$. However, for the class of matrices considered in Definition \ref{tilde} this is not the case. In fact, if $ k \leq p-1$, then since the inner product of any two distinct columns of $\Phi$ has magnitude $1/\sqrt{p}$, we have $\mu_1(k)=k\mu=\frac{k}{\sqrt{p}}$.

\end{rk}

\noindent To prove Theorem \ref{mainth}, first we will  prove the following lemmas.


\begin{lem}

\label{orderthree}

Suppose $p \geq 7$ is a prime number which is 3 mod 4.  The RIP constants $\delta_3$ of the matrix $\tilde{\Phi}$ as defined in Definition \ref{tilde}  satisfies $$\delta_3 = \sqrt{\frac{3}{p}}$$

\end{lem}

\begin{rk} 

The value of RIP constant above is consistent with the one given by  Gerhsgorin bound, i.e., $\delta_3 \leq \frac{2}{\sqrt{p}}$, and  also with the  ``improved" bound as it was given in Theorem \ref{improvedskew}, i.e., $\delta_3 \leq \frac{\frac{6}{\pi}}{\sqrt{p}}$ (note that $\sqrt{3}<6/\pi<2$). 
\end{rk} 

\begin{proof}[Proof of Lemma \ref{orderthree}] Let $G_3$ be an arbitrary Gramian matrix of order 3. We can write $G_3$ in the following form

  $$G_3=\begin{bmatrix} 1 & b & c \\ b^* & 1 & d \\ c^* & d^* & 1 \end{bmatrix} $$ where each $b,c,$ and $d$ is $\pm i/\sqrt{p}$. Let  $p(x)=\det(G_3-xI)$ be the characteristic polynomial of the matrix $G_3$. Then, we have   \begin{equation} \label{baseeqn} \begin{aligned} p ( x ) & = (1-x) \Big( (1-x)^2-dd^*\Big)-b\Big(b^*(1-x) -dc^*\Big) +c \Big( b^* d^* -c^*(1-x) \Big) \\ &=(1-x)^3-(dd^*+bb^*+cc^*)(1-x) +2 Re (bdc^*) =(1-x)^3-\frac{3(1-x)}{p} \end{aligned} \end{equation}

\noindent where we used the fact that each of $b,c$, and $d$ is $\pm \frac{i}{\sqrt{p}}$, and hence, $Re(bdc^*)=0$.  It can be easily verified that the roots of the polynomial above are $x=1-\sqrt{\frac{3}{p}}, x=1, x=1+\sqrt{\frac{3}{p}}$. Thus, \textit{all} Gramian matrices of order 3 (including the ones with the largest maximum eigenvalue, and the smallest minimum eigenvalue) have the same eigenvalues and hence, $$\delta_3=\max\{(1+\sqrt{\frac{3}{p}})-1,1-(1-\sqrt{\frac{3}{p}})\}=\sqrt{\frac{3}{p}}$$

 \end{proof}

\begin{lem} \label{aboutk} If $c \geq 1$ is a constant, then $$\frac{k-1/ck}{2}+\sqrt{\frac{(k-1/ck)^2}{4}+k+1} \leq k+1-1/c(k+1)$$ for each $k \in \mathbb{N}=\{1,2,...\}$. 

\end{lem}

\begin{proof} Note that for each $k \in \mathbb{N}$ we have $$\frac{4k+4kc-4k+4c-4}{c^2k(k+1)^2} \geq 0$$ Hence, we have $$\frac{4k-4kc(k+1)+4c(k+1)^2-4(k+1)}{c^2k(k+1)^2 } \geq 0$$ which implies $$\frac{4}{c^2(k+1)^2}-\frac{4}{c(k+1)}+\frac{4}{ck}-\frac{4}{c^2k(k+1)} \geq 0$$ Therefore, $$\frac{4}{c^2(k+1)^2}-\frac{4}{c}(1-\frac{1}{k+1})+\frac{2}{c}-\frac{8}{c(k+1)}+\frac{4}{ck}-\frac{4}{c^2k(k+1)} \geq -\frac{2}{c}$$ Hence, $$\begin{aligned} k^2 & +4+\frac{4}{c^2(k+1)^2}+\frac{1}{c^2k^2} +4k-\frac{4k}{c(k+1)}+\frac{2}{c}-\frac{8}{c(k+1)}+\frac{4}{ck}-\frac{4}{c^2k(k+1)} \\ &  \geq k^2+\frac{1}{c^2k^2}-\frac{2}{c}+4k+4 \end{aligned} $$ Thus, $$(k+2-\frac{2}{c(k+1)}+\frac{1}{ck})^2 \geq (k-1/ck)^2+4k+4$$ Hence, $$\Big( \frac{k+2-\frac{2}{c(k+1)}+\frac{1}{ck}}{2} \Big)^2 \geq \frac{(k-1/ck)^2}{4}+k+1 $$ On the other hand, it is obvious that the expression whose square we compute on the left hand side is positive, i.e., $k+2-\frac{2}{c(k+1)}+\frac{1}{ck} >0 $, therefore, $$k+1-1/c(k+1)-\frac{k-1/ck}{2} \geq \sqrt{ \frac{(k-1/ck)^2}{4}+k+1}$$ which implies the Lemma. \end{proof}

\noindent Next, we give the proof for Theorem \ref{mainth}. 

\begin{proof}[Proof of Theorem \ref{mainth}] 

To prove this theorem, we use  induction on $k$. Let $p \equiv 3 \mod 4$ be a prime satisfying the condition of the theorem, and let $\tilde{\Phi}$ be the $(p+1)/2 \times p $ measurement matrix as defined in Definition \ref{tilde}. Let $\Gamma \subseteq \{ 1,2,..., p  \}$, and by $\tilde{\Phi}_{\Gamma}$ we mean an  $m \times |\Gamma| $ matrix defined by restriction of $\tilde{\Phi}$ to the columns indexed by the elements of the set $\Gamma$. Define $\lambda_{k} ^{min} $ and $\lambda_{k} ^{max}$ as 
 
\begin{equation} \begin{aligned} \lambda_{k}^{max} &=\max_{\Gamma: | \Gamma | \leq k}\limits{\lambda_{max} (\Phi_\Gamma^* \Phi_\Gamma)}=\lambda_{\max}(\Phi_{\Gamma_0}^* \Phi_{\Gamma_0}), \\ \lambda_{k}^{min} &=\min_{\Gamma: | \Gamma | \leq k}\limits{\lambda_{min} (\Phi_\Gamma^* \Phi_\Gamma)}= \lambda_{\min}(\Phi_{\Gamma_1}^* \Phi_{\Gamma_1}) , \end{aligned} \end{equation} 
where $\Gamma_0$ and $\Gamma_1$ are sets with $k$ elements,  $\lambda_{max}(A)$ and $\lambda_{min}(A)$ denote the maximum and minimum eigenvalues of a matrix $A$ respectively, and $G_k^{\max} :=\Phi_{\Gamma_0}^* \Phi_{\Gamma_0}$ and $G_k^{\min} :=\Phi_{\Gamma_1}^* \Phi_{\Gamma_1}$ denote the Gramian matrices with largest maximum eigenvalue and smallest minimum eigenvalue respectively. Note that the RIP constant $\delta_k$  of $\tilde{\Phi}$ is given by 
\begin{equation}
\label{delta}
\delta_k=\max \{1-\lambda_{k}^{min},  \lambda_{k}^{max} -1\}
\end{equation}

 
 In order to prove the theorem, we find an upper bound for the maximum eigenvalue of $G_k^{\max}$ and a lower bound for the minimum eigenvalue of $G_k^{\min}$ respectively.   
 
 \noindent \textbf{Proving $\lambda_{\max}(G_k^{\max}) \leq 1+\frac{k-1-\frac{1}{c(k-1)}}{\sqrt{p}}$ for $k \geq 3$}: First, note that the result holds for  $k=3$ by Lemma \ref{orderthree}. Indeed by this lemma, $\lambda_{\max}(G_3^{\max}) \leq 1+\frac{\sqrt{3}}{\sqrt{p}} =1+ \frac{2- \frac{1}{2c}}{\sqrt{p}}$ for $c$ satisfying $\frac{1}{2c} =2-\sqrt{3}$, i.e., $c=\frac{1}{4-2\sqrt{3}}$.      \\
 
 \noindent Now assume that the statement is valid for $k$, then $\lambda_{k+1}^{max} \leq 1+ \frac{k-1/ck}{\sqrt{p}}$ . We will show that $ \lambda_{k+2}^{max} \leq 1+ \frac{k+1-1/c(k+1)}{\sqrt{p}}$. \\

\noindent To bound $\lambda_{k+2}^{\max}$ in terms of $\lambda_{k+1}^{\max}$, we will use the Dembo bounds as stated in Theorem \ref{dembobounds}. In particular, if $R$ is a positive semidefinite Hermitian matrix such that  $R=\begin{bmatrix} c & \textbf{b}^* \\ \textbf{b} & Q \end{bmatrix}$, $Q$ is a $(k+1) \times (k+1)$, positive semidefinite Hermitian matrix, and $\lambda_1 \leq \lambda_2 \hdots \leq \lambda_{k+2}$ are eigenvalues of $R$, then  \begin{equation} \label{dembo1} \lambda_{k+2} \leq \frac{c+\eta_{k+1}}{2}+\sqrt{\frac{(c-\eta_{k+1})^2}{4}+b^{*}b} \end{equation} and \begin{equation} \label{dembo2}  \lambda_1 \geq \frac{c+\eta_1}{2}-\sqrt{\frac{(c-\eta_1)^2}{4}+b^{*}b }  \end{equation} 

\noindent for any $\eta_1$ and $\eta_{k+1}$ such that $\lambda_{max}(Q) \leq \eta_{k+1}$ and $\lambda_{min}(Q) \geq \eta_1$. In our case, $Q$ is the matrix $\tilde{\Phi}_{\Gamma'}^* \tilde{\Phi}_{\Gamma'}$, and $R$ is the matrix $\tilde{\Phi}_{\Gamma}^* \tilde{\Phi}_{\Gamma}$, where $\Gamma, \Gamma' \subseteq \{1,2,...,p \}$ with $|\Gamma'|=k+1$ and $\Gamma \supseteq \Gamma'$ with $|\Gamma|=k+2$. Say, $\Gamma'=\{j_2,j_3,...,j_{k+2} \}$, and $\Gamma=\{j_1,j_2,...,j_{k+2} \}$.  Then $\tilde{\Phi}_{\Gamma}=[\phi_{j_1}, \phi_{j_2}, ..., \phi_{j_{k+2}}]$, and $$\tilde{\Phi}_{\Gamma}^* \tilde{\Phi}_{\Gamma}=\begin{bmatrix} \langle \phi_{j_1}, \phi_{j_1} \rangle & \langle \phi_{j_1}, \phi_{j_2} \rangle & \hdots & \langle \phi_{j_1}, \phi_{j_{k+2}} \rangle \\   \langle \phi_{j_2}, \phi_{j_1} \rangle & \langle \phi_{j_2}, \phi_{j_2} \rangle  & \hdots & \langle \phi_{j_2}, \phi_{j_{k+2}} \rangle \\ \vdots \\   \langle \phi_{j_{k+2}}, \phi_{j_1} \rangle & \langle \phi_{j_{k+2}}, \phi_{j_2} \rangle  & \hdots  & \langle \phi_{j_{k+2}}, \phi_{j_{k+2}} \rangle  \end{bmatrix}$$

\noindent  so we have $c=1$ (as $\|\phi_j\|=1$ for all $j$), $\textbf{b}=[ \langle \phi_{j_2}, \phi_{j_1} \rangle \ \langle \phi_{j_3}, \phi_{j_1} \rangle \ \hdots \ \langle \phi_{j_{k+2}}, \phi_{j_{1}} \rangle]^T$, and $$Q=\begin{bmatrix} \langle \phi_{j_2}, \phi_{j_2} \rangle & \langle \phi_{j_2}, \phi_{j_3} \rangle & \hdots & \langle \phi_{j_2}, \phi_{j_{k+2}} \rangle \\   \langle \phi_{j_3}, \phi_{j_2} \rangle & \langle \phi_{j_3}, \phi_{j_3} \rangle  & \hdots & \langle \phi_{j_3}, \phi_{j_{k+2}} \rangle \\ \vdots \\   \langle \phi_{j_{k+2}}, \phi_{j_2} \rangle & \langle \phi_{j_{k+2}}, \phi_{j_3} \rangle  & \hdots  & \langle \phi_{j_{k+2}}, \phi_{j_{k+2}} \rangle  \end{bmatrix}$$

\noindent Also, by induction hypothesis we have $\eta_{k+1} \leq 1+\frac{k-1/ck}{\sqrt{p}}$. Hence \eqref{dembo1} implies that,  


$$\lambda_{k+2} \leq 1 + \frac{k-1/ck}{2\sqrt{p}}+\sqrt{\frac{(k-1/ck)^{2 }}{4 p}+\textbf{b}^*\textbf{b}}$$ On the other hand, using the fact that each entry of $\textbf{b}$ is $\pm \frac{i}{\sqrt{p}}$ we have $\textbf{b}^*\textbf{b} = \sum_{i=1}^{k+1} \frac{1}{p}=\frac{k+1}{p}$. Therefore, to prove \begin{equation} \label{toprovelambda} \lambda_{k+2} \leq 1+ \frac{k+1-\frac{1}{c(k+1)}}{\sqrt{p}}, \end{equation} it is enough to prove : 

 \begin{equation} \label{enoughtoprove} \frac{k-1/ck}{2\sqrt{p}}+\sqrt{\frac{(k-1/ck)^2}{4 p}+\frac{k+1}{p}} \leq \frac{k+1-1/c(k+1)}{\sqrt{p}} \end{equation}

 \noindent This inequality holds by Lemma \ref{aboutk}. Next, we observe that to calculate $\lambda_{k+2}^{\max}$, we have to consider all such matrices $R$ as mentioned above and take maximum over all such choices. In other words, $$\lambda_{k+2}^{\max} =\max_{\Gamma} \lambda_{k+2}(\tilde{\Phi}_{\Gamma}^* \tilde{\Phi}_{\Gamma})$$ 
 However, as it was seen in \eqref{toprovelambda}, the value of $\lambda_{k+2}(\tilde{\Phi}_{\Gamma}^* \tilde{\Phi}_{\Gamma}) $ only depends on $|\Gamma|=k+2$, and not the elements of $\Gamma$. Therefore, the same upper bound holds for $\lambda_{k+2}^{\max}$, i.e., $$\lambda_{k+2}^{\max} \leq 1+ \frac{k+1-\frac{1}{c(k+1)}}{\sqrt{p}}$$

 \noindent \textbf{Proving  $\lambda_{k}^{\min} \geq 1-\frac{k-1-\frac{1}{c(k-1)}}{\sqrt{p}}$ for $k \geq 3$}: Similar to the argument given above, the induction base holds by Lemma \ref{orderthree}.  Assuming that the statement is valid for $(k+1)$ (induction hypothesis), we prove it for $(k+2)$. Using the same notation used above, and using Dembo bound \eqref{dembo2}, we can find a lower bound for the minimum eigenvalue of $R$, $\lambda_1(R)$, as follows.  $$\lambda_{1} \geq   1- \frac{k-1/ck}{2\sqrt{p}}-\sqrt{\frac{(k-1/ck)^{2 }}{4 p}+\textbf{b}^*\textbf{b}} =  1- \frac{k-1/ck}{2\sqrt{m}}-\sqrt{\frac{(k-1/ck)^2}{4 p}+\frac{k+1}{p}} $$

\noindent where we used the fact that $\textbf{b}^*\textbf{b} = \frac{k+1}{p}$. Hence, using  \eqref{enoughtoprove}, we obtain \begin{equation} \label{toprove2} \lambda_{1} \geq 1- \frac{k+1-1/c(k+1)}{\sqrt{p}} \end{equation}

\noindent Again, note that $\lambda_{k+2}^{\min}$ can be calculated by considering all such matrices $R$ and taking a minimum over them, i.e.,

$$\lambda_{k+2}^{\min} =\min_{\Gamma} \lambda_{1}(\tilde{\Phi}_{\Gamma}^*\tilde{\Phi}_{\Gamma})$$ 
 \noindent and as seen in \eqref{toprove2}, the value of $ \min_{\Gamma} \lambda_{1}(\tilde{\Phi}_{\Gamma}^*\tilde{\Phi}_{\Gamma})$ depends only on   $|\Gamma|=k+2$, and not the elements of $\Gamma$. Therefore, the same lower bound holds for $\lambda_{k+2}^{\min}$, i.e., $$\lambda_{k+2}^{\min} \geq 1- \frac{k+1-\frac{1}{c(k+1)}}{\sqrt{p}}$$

\noindent Therefore, $$\delta_{k+2} =\max\{\lambda_{k+2}^{max}-1,1-\lambda_{k+2}^{min}\} \leq \frac{k+1-1/c(k+1)}{\sqrt{p}}$$ \noindent as desired. 

\end{proof}

\section{A generalized Dembo approach}
\label{generalizeddembo}

Throughout this paper, we have focused so far on the bounds on the maximum and minimum eigenvalues of a Hermitian positive semidefinite matrix $A$ given by Dembo bounds as stated in Theorem \ref{dembobounds}. These bounds are obtained by considering the maximum and minimum eigenvalues of $2 \times 2$ block matrices $R_1$ and $R_2$ satisfying $R_1 \geq A$, and $A \leq R_2$ respectively. In this section, our goal is to tighten these bounds by following a similar idea. In particular,  we would like to consider the maximum and minimum eigenvalues of $3 \times 3$ block matrices $Q_1$ and $Q_2$ satisfying $Q_1 \geq A$, and $A \leq Q_2$, in order to obtain bounds on the extreme eigenvalues of $A$. \\


\begin{lem} \label{newboundlem} Suppose that a $(k+2) \times (k+2)$ positive semidefinte Hermitian matrix $R$ can be written as $ R=\begin{bmatrix} a & b & \textbf{c} \\ b^*  & a &  \textbf{d} \\ \textbf{c}^* & \textbf{d}^* & Q  \end{bmatrix}$ where $Q$ is a $k \times k$ positive semidefinite Hermitian matrix, and $\textbf{c}$ and $\textbf{d}$ are $k \times 1$ vectors, and $a,b \in \mathbb{C}$.  Also, suppose that $ \lambda_1 \leq \lambda_2 \leq ... \leq \lambda_{k+1} \leq \lambda_{k+2}$ are the eigenvalues of $R$. Then $\lambda_{k+2} \leq \nu_{max}$ and $\lambda_1  \geq \nu_{min}$ where $\nu_{max}$ and $\nu_{min}$ are the maximum and minimum roots of  characteristic polynomials of $ R_1=\begin{bmatrix} a & b & \textbf{c} \\ b^*  & a &  \textbf{d} \\ \textbf{c}^* & \textbf{d}^* & \eta_k I  \end{bmatrix}$ and $ R_2=\begin{bmatrix} a & b & \textbf{c} \\ b^*  & a &  \textbf{d} \\ \textbf{c}^* & \textbf{d}^* & \eta_1 I  \end{bmatrix}$ respectively. Here, as before,  $\eta_1$ is any lower bound on the minimum eigenvalue of $Q$ and $\eta_k$ is any upper bound on the maximum eigenvalue of $Q$. 
\end{lem}

\begin{proof} Since $Q - \eta_1 I \geq 0$, we have $R-R_2 \geq 0$, and so $\lambda_{min}(R) \geq \lambda_{min} (R_2)$. Similarly, $\eta_k I -Q \geq 0$, implies that $\lambda_{max}(R) \leq \lambda_{max} (R_1)$.

\end{proof}

\noindent Next, to estimate the extreme eigenvalues of matrices of the form $R_1$ or $R_2$ mentioned above, we need to obtain a formula for determinant of these matrices.  To that end, we use the so called \textit{Schur determinant formula}.

\begin{lem} \cite[p. 50]{schurlem}  (Schur determinant formula) Let $M$ be a $2 \times 2$ block matrix of the form $$M=\begin{bmatrix} P & Q \\ R & S \end{bmatrix}$$ where $P$ is a $p \times p$ matrix, $S$ is an $s \times s$ matrix, $Q$ is a $p \times s$ matrix, and $R$ is an $s \times p$ matrix. If $P$ is invertible, then $$\det (M)= \det (P) \cdot \det(S-RP^{-1}Q) $$ Similarly, if $S$ is invertible, then $$\det(M)=\det (S) \det(P-Q S^{-1} R) $$

\end{lem}

\noindent An immediate corollary of the lemma above is that for a $2 \times 2$ block matrix of the form $$R=\begin{bmatrix} a & \textbf{b} \\ \textbf{c} & \eta I_k \end{bmatrix} $$ where $\eta \ne 0$, and $\textbf{b}$, and $\textbf{c}$ are $1 \times k$, and $k \times 1$ vectors respectively, we have $$\det (R)=\eta^k ( a-\textbf{b}\eta^{-1} \textbf{c} ) $$

\noindent Now, we are ready to derive a  formula regarding the determinant of the $3 \times 3$ block matrices of the form mentioned above.

\begin{lem}
\label{aboutalpha}
Let $R$ be a $(k+2) \times (k+2)$ matrix that can be written of the form $R=\begin{bmatrix} a & b & \textbf{c} \\  b^* & a & \textbf{d} \\ \textbf{c}^* & \textbf{d}^* & \eta I_k \end{bmatrix}$, where $\textbf{c}=(c_1,c_2,...,c_k)$, $\textbf{d}=(d_1,d_2,...,d_k)$ are $1 \times k$ vectors. For each $1 \leq i \leq k$, define the vectors $\textbf{c}_i$, and $\textbf{d}_i$ as $1 \times (k-1)$ vectors obtained by removing the $i$th entry  $c_i$ and $d_i$ from the vectors $\textbf{c}$ and $\textbf{d}$ respectively. Then we have $$\det(R)= \eta^{k-2} \Big( a^2 \eta^2 -a \eta (dd^*+ cc^*)  -bb^* \eta^2+2Re( b\textbf{d}\textbf{c}^*) +\gamma \Big)   $$

\noindent where $\gamma$ is defined as  $\gamma:=   \sum_{i=1}^k |c_i|^2 \textbf{d}_i \textbf{d}_i^* -\sum_{i=1}^k c_i d_i^* \textbf{d}_i\textbf{c}_i^* $.

\end{lem}

Note that the determinant of a Hermitian matrix must be a real number. Now, apart from the term $\sum_{i=1}^k c_i d_i^* \textbf{d}_i\textbf{c}_i^* $, other terms are obviously real. This term is also real because if say $c_{\ell}, d_{\ell} \ne 0$, then the term $c_{\ell} d_{\ell}^* \textbf{d}_{\ell}\textbf{c}_{\ell}^*$ will be appeared in the sum above. Now, if $c_{m}, d_{m} \ne 0$, for some $m \ne \ell$, then $c_m^* d_m$ will be one of the terms appearing in the expansion of $\textbf{d}_{\ell}\textbf{c}_{\ell}^*$. Hence, $c_{\ell}d_{\ell}^* c_m^* d_m$ will be a generic non-zero term of the expansion of $c_{\ell} d_{\ell}^*\textbf{d}_{\ell}\textbf{c}_{\ell}^*$. Next, note that this term will be accompanied by $c_{m}d_{m}^* c_{\ell}^* d_{\ell}$ which is a generic term in the expansion of $c_{m} d_{m}^*\textbf{d}_{m}\textbf{c}_{m}^*$. Thus, every term in this sum will be accompanied by its complex conjugate, and hence, this term is also real.

\begin{proof}[Proof of Lemma \ref{aboutalpha}] Expanding the determinant along the first row we obtain

$$ \begin{aligned} \det(R)& =a \det \begin{bmatrix} a & \textbf{d} \\ \textbf{d}^* & \eta I_k \end{bmatrix} -b \det \begin{bmatrix} b^* & \textbf{d} \\ \textbf{c}^* & \eta I_k \end{bmatrix}+c_1 \det \begin{bmatrix} b^* & a & d_2 & d_3 & ... & d_k   \\ c_1^*  & d_1^* &  0 & 0 & ... & 0 \\ c_2^* &  d_2^* & \eta & 0 & ... & 0  \\ \vdots \\ c_k^* & d_k ^* & 0 & 0 &  ... & \eta \end{bmatrix} -\\ & - c_2 \det \begin{bmatrix} b^* & a & d_1 & d_3 & \hdots & d_k \\ c_1^* & d_1^* & \eta & 0 &  \hdots & 0 \\ c_2^* & d_2^* & 0  & 0 & \hdots & 0 \\c_3^* & d_3^* & 0 & \eta & ... & 0 \\ \vdots \\ c_k^* & d_k^* & 0 & 0 & \hdots  & \eta \end{bmatrix} + \hdots \\ & + (-1)^{k+1} c_k\begin{bmatrix} b^* & a & d_1 & d_2 & d_3 & \hdots & d_{k-1} \\ c_1^* & d_1^* & \eta & 0 & 0 &  \hdots & 0 \\ c_2^* & d_2^* & 0  & \eta & 0 & \hdots & 0 \\c_3^* & d_3^* & 0 & 0 &  \eta & ... & 0 \\ \vdots \\ c_k^* & d_k^* & 0 & 0 & 0  & \hdots  & 0 \end{bmatrix}  \end{aligned} $$ 

\noindent For the first two terms above, we use Schur determinant formula, and we expand the remaining terms along the rows with more number of zeros. $$ \begin{aligned} \det(R) &= a  \eta^k (a - \eta^{-1} \textbf{d} \textbf{d}^* ) -b \eta^k ( b^*- \eta^{-1} \textbf{d}\textbf{c}^*) +c_1 \Big(- c_1^* \det \begin{bmatrix} a & \textbf{d}_1 \\ \textbf{d}_1^* & \eta I_{k-1} \end{bmatrix} \\ & +d_1^* \begin{bmatrix} b^* & \textbf{d}_1 \\ \textbf{c}_1^* & \eta I_{k-1} \end{bmatrix} \Big) - c_2 \Big( c_2^* \begin{bmatrix} a & \textbf{d}_2 \\ \textbf{d}_2^* & \eta I_{k-1} \end{bmatrix} -d_2^* \det \begin{bmatrix} b^* & \textbf{d}_2 \\ \textbf{c}_2^* & \eta I_{k-1} \end{bmatrix} \Big) + \hdots  \\ & + (-1)^{k+1}  c_k \Big( (-1)^{k+2} c_k^* \begin{bmatrix} a & \textbf{d}_k \\ \textbf{d}_k^* & \eta I_{k-1} \end{bmatrix} +(-1)^{k+3} d_k^* \det \begin{bmatrix} b^* & \textbf{d}_k \\ \textbf{c}_k^* & \eta I_{k-1} \end{bmatrix} \Big) \end{aligned} $$ 

\noindent Next, we use the Schur determinant formula to expand the determinant of $2 \times 2$ block matrices above. $$\begin{aligned} \det(R) & = a^2 \eta^k -a \eta^{k-1} \textbf{d} \textbf{d}^* -b b^* \eta^k + b \eta^{k-1} \textbf{d}\textbf{c}^*  - \\ & \Big( c_1c_1^* \eta^{k-1} a -\eta^{k-2} c_1c_1^* \textbf{d}_1 \textbf{d}_1^*  +c_2c_2^* \eta^{k-1} a -\eta^{k-2} c_2 c_2^*\textbf{d}_2\textbf{d}_2^*+...\\ & +c_kc_k^* \eta^{k-1}a -\eta^{k-2} c_k c_k^* \textbf{d}_k \textbf{d}_k^* \Big)  \\ & + c_1d_1 b^* \eta^{k-1}- c_1d_1 \textbf{d}_1\textbf{c}_1^* \eta^{k-2} + \hdots + c_k d_k^* b^* \eta^{k-1} -c_k d_k^* \textbf{d}_k\textbf{c}_k^* \eta^{k-2} \\ & = \eta^{k-2} \Bigg( a^2 \eta^2 -a\eta \textbf{d}\textbf{d}^* -bb^* \eta^2 +b \eta \textbf{d}\textbf{c}^* -a \textbf{c} \textbf{c}^* \eta+\sum_{i=1}^k c_ic_i^* \textbf{d}_i\textbf{d}_i^* \\ & + \eta b^* \textbf{c} \textbf{d}^* -\sum_{i=1}^k c_i d_i^* \textbf{d}_i\textbf{c}_i^* \Bigg) \\ &=  \eta^{k-1} \Bigg( a^2 \eta -a \textbf{d}\textbf{d}^* -bb^* \eta  +2 Re(b  \textbf{d}\textbf{c}^*) -a \textbf{c} \textbf{c}^* \Bigg) \\& +\eta^{k-2} \Bigg( \sum_{i=1}^k c_ic_i^* \textbf{d}_i\textbf{d}_i^* -\sum_{i=1}^k c_i d_i^* \textbf{d}_i\textbf{c}_i^* \Bigg) \\ &= \eta^{k-2} \Big( a^2 \eta^2 -a \eta (\textbf{d}\textbf{d}^*+ \textbf{c}\textbf{c}^*)  -bb^* \eta^2+2Re( b\textbf{d}\textbf{c}^*) +\alpha \Big) \end{aligned}  $$

\noindent where $\gamma$ is the expression defined as  $\gamma:=  \Bigg( \sum_{i=1}^k c_ic_i^* \textbf{d}_i\textbf{d}_i^* -\sum_{i=1}^k c_i d_i^* \textbf{d}_i\textbf{c}_i^* \Bigg)$. \end{proof}

\noindent Using the lemmas proved above,  we show that the RIP constant of the Paley CS matrices-- as defined in Definition \ref{tilde}-- can be improved where the improvement term is  a universal constant, unlike the situation in the previous section (see Theorem \ref{mainth}), where the improvement term was dependent on the sparsity level.   

\begin{Theor} 
\label{mainthgeneral}

Let $p \geq 7$ be a prime number such that $p \equiv 3 \mod 4$.   The RIP constants of the measurement matrix $\tilde{\Phi}$ as given in Definition \ref{tilde} satisfies $$\delta_k \leq \frac{k-1-\frac{2}{3}(2-\sqrt{3})}{\sqrt{p}} $$

\end{Theor}

\begin{proof} The idea of the proof is similar to the one given for Theorem \ref{mainth}, and we use  similar notation. Hence, we assume that the result holds for $k$ which gives an upper bound and a lower bound for the eigenvalues of a Gramian matrix of the size $k \times k$, and we prove the statement for $(k+2)$. Therefore, the proof includes two main steps, one regarding the maximum eigenvalue, and one regarding the minimum eigenvalue. \\

\begin{enumerate}
\item

 We will prove that $\lambda_{\max}(G_k^{\max}) \leq 1+ \frac{k-1-\frac{2}{3}(2-\sqrt{3}-\alpha)}{\sqrt{p}}$ for $k\geq 3$ using induction. To that end, we first verify the statement for $k=3$ and $k=4$; then we finish by assuming it holds for $k$, and proving that this implies it holds for $(k+2)$. The induction base ($k=3$) holds by Lemma \ref{orderthree}, since by this lemma, $$\lambda_3^{\max} \leq 1+\frac{\sqrt{3}}{\sqrt{p}}=1+\frac{2-(2-\sqrt{3})}{\sqrt{p}} \leq 1+ \frac{2-\frac{2}{3} (2-\sqrt{3})}{\sqrt{p}}. $$ 
\noindent  The other induction base ($k=4$) also holds by Theorem \ref{mainth}. Setting  $k=4$ in this theorem, we obtain $\delta_4 \leq \frac{3-\frac{1}{3c}}{\sqrt{p}}$, which implies $$\lambda_4^{\max} \leq 1+\frac{3-\frac{1}{3c}}{\sqrt{p}}=1+\frac{3-\frac{2}{3}\frac{1}{2c}}{\sqrt{p}}=1+\frac{3-\frac{2}{3}(2-\sqrt{3})}{\sqrt{p}} $$

\noindent Next, consider $G_{k+2}^{\max}$, the $(k+2) \times (k+2)$ matrix obtained from the Gramian matrix indexed by the set $\Gamma=\{r_1,r_2,...,r_{k+2}\} \subseteq \{1,2,...,p \}$ (with $|\Gamma|=k+2$,) and we write  it in the following form.  

\begin{equation} \label{gkmax} G_{k+2}^{\max}=\begin{bmatrix} 1 & b & \textbf{c} \\ b^*  & 1 &  \textbf{d} \\ \textbf{c}^* & \textbf{d}^* & Q  \end{bmatrix} \end{equation}

 Here, $b=\langle \tilde{\Phi}_{r_1},\tilde{\Phi}_{r_2} \rangle$, $\textbf{c}=(c_1,...,c_k)$, $\textbf{d}=(d_1,...,d_k)$, $c_i=\langle \tilde{\Phi}_{r_1},\tilde{\Phi}_{r_{i+2}} \rangle $, $d_i=\langle \tilde{\Phi}_{r_2}, \tilde{\Phi}_{r_{i+2}} \rangle $ (with $1 \leq i \leq k$). By the construction of $\tilde{\Phi}$, each non-diagonal entry of $G_{k+2}^{\max}$ (including $b,c_i$'s and $d_i$'s) is $\pm \frac{i}{\sqrt{p}}$. On the other hand, we know by Lemma \ref{newboundlem}, that the maximum eigenvalue of $G_{k+2}^{\max}$ is bounded from above by the maximum eigenvalue of   \begin{equation} \label{bkmax} B_{k+2}^{\max}=\begin{bmatrix} 1 & b & \textbf{c} \\ b^*  & 1 &  \textbf{d} \\ \textbf{c}^* & \textbf{d}^* & \eta_k I \end{bmatrix} \end{equation}  where $\eta_k$ is the upper bound on the maximum eigenvalue of $Q$ and by induction hypothesis can be written in the form $\eta_k=1+\frac{D}{\sqrt{p}}$, with $D=k-1-\frac{2}{3}(2-\sqrt{3})$. Next, let $\lambda=1+\frac{C}{\sqrt{p}}$, with $C>D+2=k+1-\frac{2}{3}(2-\sqrt{3})$. If we show that $p(\lambda) \ne 0$, where $p(x)$ is the characteristic polynomial of $B_{k+2}^{\max}$, then this shows that $1+\frac{k+1-\frac{2}{3}(2-\sqrt{3})}{\sqrt{p}}$ is an upper bound for the maximum eigenvalue of $B_{k+2}^{\max}$, and hence, for $G_{k+2}^{\max}$ as desired.  Note that $p(x)=\det(B_{k+2}^{\max}-xI)$ can be written as   $ p(x)=\det \begin{bmatrix} 1-x & b & \textbf{c} \\ b^*  & 1-x &  \textbf{d} \\ \textbf{c}^* & \textbf{d}^* & (\eta_k-x) I  \end{bmatrix}$, where each non-diagonal entry is $\pm i/\sqrt{p}$.  Following the notation of Lemma \ref{aboutk},  we have 
     
 \begin{equation} \label{generalp} \begin{aligned} p(x) & = (\eta_k-x)^{k-2} \Big( (1-x))^2 (\eta_k-x)^2 -(1-x)(\eta_k-x) (\textbf{d}\textbf{d}^*+\textbf{c}\textbf{c}^*)- \\ & (\eta_k-x)^2 bb^* +2(\eta_k-x)Re(b\textbf{d}\textbf{c}^*) +\gamma \Big)  \end{aligned} \end{equation}
 
 \noindent where  $\gamma:=   \sum_{i=1}^k |c_i|^2 \textbf{d}_i \textbf{d}_i^* -\sum_{i=1}^k c_i d_i^* \textbf{d}_i\textbf{c}_i^* $.  Next, we prove that $\gamma \geq 0$. Note that each entry of vectors $\textbf{c}$ and $\textbf{d}$ is $\pm \frac{i}{\sqrt{p}}$. Thus, $ \sum_{i=1}^k c_ic_i^* \textbf{d}_i\textbf{d}_i^* =\frac{k(k-1)}{p^2}$. Also, for each $1 \leq i \leq k$, $\textbf{d}_i \textbf{c}_i^*$ contains $(k-1)$ terms, each with magnitude $\frac{1}{p}$. Thus,  $\sum_{i=1}^k c_i d_i^* \textbf{d}_i \textbf{c}_i^*  \leq \frac{k (k-1)}{p^2} $. This implies that $\gamma \geq 0$. \\

  Furthermore, since $C>D$, we have $\eta_k-\lambda < 0$. Accordingly, to show $p(\lambda) \ne 0$, considering the fact that $\gamma \geq 0$, it is enough to show that $$q(\lambda)=(1-\lambda)^2 (\eta_k-\lambda) -(1-\lambda) (\textbf{d}\textbf{d}^*+\textbf{c}\textbf{c}^*) -(\eta_k-\lambda) bb^*  <0$$

 

%

\noindent where we used the fact that $Re(b \textbf{d} \textbf{c}^*)=0$.  Also, using $\eta_k=1+\frac{D}{\sqrt{p}}$, $\lambda=1+\frac{C}{\sqrt{p}}$, $\textbf{c}^* \textbf{c} = \frac{k}{p}$, and $\textbf{d}^*\textbf{d} = \frac{k}{p}$, we have   \begin{equation} \label{aboutp} q(\lambda ) = \frac{1}{p \sqrt{p}} \Big( C^2(D-C)+C(2k+1)-D \Big) \end{equation}

 To show that $q(\lambda)<0$ for any $C>D+2$, first  let $C=D+2=k-1/3+2\sqrt{3}/3$, and recall that $D= k-7/3+2\sqrt{3}/3$. Then,   $$ \begin{aligned} q(\lambda) & =\frac{1}{p \sqrt{p}} \Big( -2C^2+C(2k+1)-(k-7/3+2\sqrt{3}/3) \Big) \\ & =\frac{1}{p \sqrt{p}} \Big( -2(k+\frac{2 \sqrt{3}}{3}-\frac{1}{3})^2+(k+\frac{2 \sqrt{3}}{3}-\frac{1}{3})(2k+1) \\ & \ \ \ \ \  \ \ \ \ \ \ \ \ -k+\frac{7}{3}-\frac{2\sqrt{3}}{3} \Big) \\ & = \frac{1}{p \sqrt{p}} \Big( \big(-\frac{2}{3}(2\sqrt{3} -1)\big)k+2-\frac{2}{9}(2\sqrt{3}-1)^2 \Big) \\ &   \leq \frac{1}{p \sqrt{p}} \Big( \big( -\frac{2}{3}(2\sqrt{3}-1) \big)k+0.651 \Big)     < 0 \end{aligned} $$ for $k \geq 1$.  

\noindent On the other hand, if we substitute $D=k-7/3+2\sqrt{3}/3$, and subsequently differentiate the right hand side of \eqref{aboutp}, we obtain  $$\begin{aligned} \frac{d}{dC} & \Big( C^2(k-7/3+2\sqrt{3}/3-C)+C(2k+1) \Big) \\ & =-3C^2+2(k-7/3+2\sqrt{3}/3)C+2k+1  :=g(C)<0 \end{aligned} $$ for $C>k$. This is because  $g(k)=-k^2+(-8/3+4 \sqrt{3}/3) k+1<0$, and $g(C)$ is decreasing for $C>\frac{2k-14/3+\frac{4}{3}(\sqrt{3})}{6}$. Therefore, $q(\lambda)<0$ for every $\lambda=1+\frac{C}{\sqrt{p}}$ and $C> k-1/3+2\sqrt{3}/3$. Hence, an upper bound for the maximum eigenvalue of $(k+2) \times (k+2)$ matrix $R_2$ (and hence for $G_{k+2}^{\max}$) is $\lambda=1+\frac{k-1/3+2\sqrt{3}/3}{\sqrt{p}}$, as desired. Note this bound only depends on $|\Gamma|$, and not the elements of $\Gamma$.

\item We will prove that  $\lambda_{\min}(G_k^{\min}) \geq 1- \frac{k-1-\frac{2}{3}(2-\sqrt{3})}{\sqrt{p}}$ for $k\geq 3$ using induction. The proof is similar to above. For the sake of completeness, we state it briefly.  The induction base ($k=3$) holds by Lemma \ref{orderthree}, since by this lemma, $$\lambda_3^{\min} \geq 1-\frac{\sqrt{3}}{\sqrt{p}} =1-\frac{2-(2-\sqrt{3})}{\sqrt{p}} \geq 1- \frac{2-\frac{2}{3}(2-\sqrt{3})}{\sqrt{p}}$$
\noindent  The other induction base ($k=4$) also holds by Theorem \ref{mainth}: Set $k=4$ in this theorem. Then we obtain $\delta_4 \leq \frac{3-\frac{1}{3c}}{\sqrt{p}}$, which implies $$\lambda_4^{\min} \geq 1-\frac{3-\frac{1}{3c}}{\sqrt{p}}=1-\frac{3-\frac{2}{3}\frac{1}{2c}}{\sqrt{p}}=1-\frac{3-\frac{2}{3}(2-\sqrt{3})}{\sqrt{p}} $$

\noindent Next, consider the $(k+2) \times (k+2)$ Gramian matrix $G_{k+2}^{\min}$.  We write this matrix of the form given in \eqref{gkmax}, and we consider the matrix $B_{k+2}^{\min}$ of the form given in \eqref{bkmax}, and with $\eta_k$ replaced by $\eta_1$, namely, the lower bound for the minimum eigenvalue of $Q$. We write $\eta_1$ of the form  $\eta_1=1-D/\sqrt{p}$, and we consider $\lambda:=1-C/\sqrt{p}$ with $C>D+2$. We will consider $p(x)$, the characteristic polynomial of $B_{k+2}^{\min}$, and will show that $p(\lambda) \ne 0$. To do so, it is enough to show that

$$q(\lambda)=(1-\lambda)^2 (\eta_k-\lambda) -(1-\lambda) (\textbf{d}\textbf{d}^*+\textbf{c}\textbf{c}^*) -(\eta_k-\lambda) bb^*  >0$$

 The expression for $q(\lambda)$ can be simplified as 


   \begin{equation} \label{aboutp2} q(\lambda) = \frac{1}{p \sqrt{p}} \Big( C^2(C-D)-C(2k+1)+D \Big) \end{equation}

 Now,  let $C=k-1/3+2\sqrt{3}/3$, and recall that $D= k-7/3+2\sqrt{3}/3$. Next,   $$ \begin{aligned} q(\lambda) & = \frac{1}{p \sqrt{p}} \Big( 2C^2-C(2k+1)+(k-7/3+2\sqrt{3}/3) \Big) \\ &  \geq \frac{1}{p \sqrt{p}} \Big( \big( \frac{2}{3}(2\sqrt{3}-1)\big)k-0.651 \Big)     > 0 \end{aligned} $$ for $k \geq 1$. 

\noindent On the other hand, if we substitute $D=k-7/3 +2\sqrt{3}/3$, and subsequently differentiate the right hand side of \eqref{aboutp2}, we will get  $$\begin{aligned} \frac{d}{dC} & \Big( C^2(-k+7/3-2\sqrt{3}/3+C)-C(2k+1) \Big) \\ & =3C^2-2(k-7/3+2\sqrt{3}/3)C-2k-1  :=g(C)>0 \end{aligned} $$ for $C>k$. This is because  $g(k)=k^2+(8/3-4 \sqrt{3}/3) k-1>0$, and $g(C)$ is increasing for $C>\frac{2k-14/3+\frac{4}{3}(\sqrt{3})}{6}$. Therefore, $q(\lambda)>0$ for every $\lambda=1-\frac{C}{\sqrt{p}}$ and $C > k-1/3+2\sqrt{3}/3$. Hence, a lower bound for the minimum eigenvalue of $(k+2) \times (k+2)$ matrix $R_2$ (and hence for $G_{k+2}^{\min}$) is $\lambda=1-\frac{k-1/3+2\sqrt{3}/3}{\sqrt{p}}$, as desired. Note that this bound also only depends on $|\Gamma|$, and not the elements of $\Gamma$.

  Considering the calculations we did for the maximum and minimum eigenvalues of the Gramian matrices $R_1$ and $R_2$, we conclude that the RIP constant of order $(k+2)$ satisfies  $$\delta_{k+2} \leq \frac{k-1/3+2\sqrt{3}/3}{\sqrt{p}}$$ proving the theorem using induction. 
  
  \end{enumerate}
  
   \end{proof}

 \section{A path to break the square-root barrier using Dembo bounds} 
 \label{pathtobreak}

 In this section, we propose an approach that can lead to breaking the square-root barrier for the construction given in Definition \ref{tilde}, if a specific conjecture regarding the distribution of quadratic residues holds. Our approach is based on the generalized Dembo bounds as derived and explained in section \ref{generalizeddembo}. Let $\Phi$ denote the measurement matrix as defined in Definition \ref{tilde}. We saw in Section  \ref{specificmatrix} that if all upper diagonal entries of the Gramian matrix $G=\Phi_T^* \Phi_T$, corresponding to the index set $T=\{r_1,r_2,...,r_k \}$, are $i/\sqrt{p}$ (and all the lower elements are $-i/\sqrt{p}$) then we would have a multiplicative improvement for Gershgorin bound but the square root barrier can not be broken. Therefore, in such case, $$ \begin{aligned} & \Big( \frac{r_1-r_3}{p} \Big) = \Big( \frac{r_1-r_4}{p} \Big)= ...=\Big( \frac{r_1-r_k}{p} \Big) =1, \\  &    \Big( \frac{r_2-r_3}{p} \Big) = \Big( \frac{r_2-r_4}{p} \Big)= ...=\Big( \frac{r_2-r_k}{p} \Big) =1      \end{aligned} $$ 
 
 \noindent If we re-tag the columns as $r'_1=r_k, r'_2=r_{k-1},...,r'_k :=r_1$, then $$ \begin{aligned} & \Big( \frac{r'_1-r'_3}{p} \Big) = \Big( \frac{r'_1-r'_4}{p} \Big)= ...=\Big( \frac{r'_1-r'_k}{p} \Big) =-1, \\  &    \Big( \frac{r'_2-r'_3}{p} \Big) = \Big( \frac{r'_2-r'_4}{p} \Big)= ...=\Big( \frac{r'_2-r'_k}{p} \Big) =-1      \end{aligned} $$ 
 
 \noindent In either case, we have $$ \sum_{i \in I} \chi(i) \chi(i+a) =|I|$$ 
 
 \noindent where $\chi (x)=\Big( \frac{x}{p} \Big)$ denotes the Legendre symbol (and hence, is a Dirichlet character), $I:=\{r_1-r_3, r_1-r_4,...,r_1-r_k\}$, and $a=r_2-r_1$. Therefore, one can hope that if the opposite to this situation occurs in the following sense, then the square-root barrier may be \textit{broken}.

 \begin{conj} \label{numth}

 There exists constants $0<\alpha<1,$ and $\nu>1/2$, and a positive integer $m_{\alpha}$ such that for any set $\{r_1,...,r_k\}$ in $\mathbb{Z}_p$, with $m_{\alpha} \leq k \leq p^{\nu}$ there exist indices $1\leq i<j \leq k$ satisfying the following inequality

   $$\frac{|\sum_{\ell \in I_{r_i,r_j}} \chi(\ell) \chi(\ell+a)|}{|I_{r_i,r_j}|} <\alpha$$
 
 \noindent where $\chi(x)=\Big(\frac{x}{p} \Big)$, $a=r_j-r_i$, and $I_{r_i,r_j} :=\{r_i-r_{\ell}: 1\leq \ell \leq k, \ell \ne i, \ell \ne j\}$. We call $I_{r_i,r_j}$ a one-sided difference set.

 \end{conj}

%
%
%
%
%
%

 Note that this conjecture is not just based on what is \textit{needed} to break the square-root barrier, but also based on a similar result already known in Number Theory. We briefly mention this result as stated in \cite{turyn}. 
 
 Let $G$ be a finite (additive) group, and let $D$ be a subset of $G$ (called a \textit{difference set}) with $k$ elements such that every non-zero element of $G$ can be uniquely written as $d_1-d_2$. Then, $$|\sum_{d \in D} \chi(d)|= \sqrt{k-1} $$
 
 \noindent Hence, for any $0< \alpha<1$, we have $$ \frac{|\sum_{d \in D} \chi(d)|}{|D|}= \frac{\sqrt{k-1}}{k}<\alpha$$ provided that $k$ is sufficiently large. We also verify this conjecture numerically with few examples.

 Set $\alpha :=0.8 $, and $m_{\alpha} :=5$. The first prime satisfying $p \geq m_{\alpha}$ is $p=5$ in which case we should start with a set with 5 elements. We only have one choice, and that is $A=\mathbb{Z}_5$ (or any permutation of it). Set $$r_1=0, \ \ r_2=1, \ \ r_3=2, \ \ r_4=3, \ \ r_5=4$$ Then, $$ \Big( \frac{r_1-r_3}{p}\Big) \Big( \frac{r_2-r_3}{p} \Big)=-1,  \Big( \frac{r_1-r_4}{p}\Big) \Big( \frac{r_2-r_4}{p} \Big)=1,  \Big( \frac{r_1-r_5}{p}\Big) \Big( \frac{r_2-r_5}{p} \Big)=-1 $$ Thus, $$\frac{|\sum_{i=3}^5 \Big(\frac{r_1-r_i}{p} \Big) \Big( \frac{r_2-r_i}{p} \Big)|}{3}=1/3<\alpha $$


\noindent We also verify for $p=19$. As we know, we should start with a set with at least 5 elements, say, we start with a set with 12 elements. We choose a random support set $T$ with 12 elements, say $T=\{8, 15, 5, 13, 10, 2, 17, 4, 11 18, 16, 19 \}$ (i.e., $r_1=8, r_2=15,...,r_{12}=19$).  Then, among 10 elements of the sequence $ \{ \Big( \frac{r_1-r_3}{p} \Big) \cdot \Big( \frac{r_2-r_3}{p}  \Big) ,...., \Big( \frac{r_1-r'_{12}}{p} \Big) \cdot \Big( \frac{r_2-r_{12}}{p} \Big) \} $, we have three -1's and seven 1's. Hence,

  $$\frac{|\sum_{i=3}^{12} \Big(\frac{r_1-r_i}{p} \Big) \Big( \frac{r_2-r_i}{p} \Big)|}{10}=0.4<\alpha$$


In the next experiment, we again work with $p=19$, but we try to consider the worst case (that could potentially fail the conjecture). Such case would occur if we choose a set $\{r_1,...,r_k\}$ such that $\Big( \frac{r_i-r_j}{p} \Big)=1$ whenever $i<j$. In particular, we can choose $T:= \{1,2,18, 16, 15, 14, 8, 7, 6, 4\}$. Then, all 10 elements of the sequence $ \{ \Big( \frac{r_1-r_3}{p} \Big) \cdot \Big( \frac{r_2-r_3}{p}  \Big) ,...., \Big( \frac{r_1-r_{12}}{p} \Big) \cdot \Big( \frac{r_2-r_{12}}{p} \Big) \} $ are 1's which is opposite to what we need. However, we can simply choose $r'_1:=r_7=8$, and $r'_2=r_8=7$, $\{r'_3,...,r'_{10} \} =T \setminus \{7,8\}$ (the order is irrelevant). Then, we would have

 $$\frac{|\sum_{i=3}^{12} \Big(\frac{r_1-r_i}{p} \Big) \Big( \frac{r_2-r_i}{p} \Big)|}{10}=\frac{2}{10}<\alpha$$
 
 \noindent Therefore, in all these experiments the conjecture was verified for $\alpha=0.8$ (in these cases, we could even choose a smaller $\alpha$, e.g., $\alpha=0.5$).

%
%

Now, based on this conjecture, we prove that the square-root barrier can be broken for the construction given in Definition \ref{tilde}. In the following, we provide a proof using induction on $k$. For the induction base, we need to use a value for the power $\beta$ in $\delta_k \leq \frac{k^{\beta}}{\sqrt{p}}$. Since numerical experiments (see Figure \ref{aboutj}) suggests $\beta<0.7$, we use $\beta=0.7$ as it seems that this is the value that works for any $k$-value. 

\begin{prop}\label{squarerootb} (breaking the square-root barrier).  Suppose Conjecture \ref{numth} holds with $\alpha, \nu, m_{\alpha}$ as defined in the statement of this conjecture.  Let  $\beta=0.7$, and let $c_{\alpha}$ be a fixed integer such that  the  following inequality holds: $$12c_{\alpha}^{1+\beta} < (1-\alpha) c_{\alpha}^2-2c_{\alpha},$$


  \noindent Also, let $b_{\alpha}=\max\{m_{\alpha}+2,c_{\alpha}+2\} $, and suppose that  $p \geq b_{\alpha}^2$. Then, for the construction given in Definition \ref{tilde}, and for $k \leq \min\{\frac{p^{\nu}}{2},\frac{p^{\frac{1}{2\beta}}}{2}\}$, we have $$\delta_k < 1/\sqrt{2}$$  Hence,  the square-root barrier would be broken for this construction. 

\end{prop}

\begin{rk} Here, we make the above statement more concrete. Set $\alpha=0.8, m_{\alpha}=5$, and $\nu=0.8$. Then, $c_{\alpha}$ is the smallest integer satisfying $$12x^{1.7} \leq 0.2x^2-2x$$ Numerically we check that we can set  $c_{\alpha}=899,998$, and this gives  $b_{\alpha}=900,000$. So the proposition above then reduces to: 

``Suppose Conjecture \ref{numth} holds with the values mentioned above, and let $p \equiv 3$ mod 4 be a prime number satisfying $p \geq 81 \times 10^{10}$, and consider the construction given in Definition \ref{tilde}. Then, the RIP constants of such construction satisfy $\delta_{2k} < 1/\sqrt{2}$ for any $k<\frac{p^{5/7}}{2}$".
\end{rk}



 \begin{proof}[Proof of Proposition \ref{squarerootb}] First, let $k \leq b_{\alpha}$. Then, $k \leq \sqrt{p}$, and hence, by Theorem \ref{improvedskew}, the RIP constant of the measurement matrix $\tilde{\Phi}$ satisfies $$\delta_k \leq \frac{2}{\pi} \frac{k}{\sqrt{p}} \leq \frac{2}{\pi}<\frac{1}{\sqrt{2}} $$ as desired. Next, let $k=k_0 \geq b_{\alpha}$ be an arbitrary integer (with $k_0 \leq p^{\nu}$), and we prove that \begin{equation} \label{deltakbet} \delta_{k} \leq \frac{k^{\beta}}{\sqrt{p}} \end{equation} holds for $k=k_0$ using induction. To do that, first verify \eqref{deltakbet} for $k=b_{\alpha}-1, b_{\alpha}-2$ numerically (induction base). Now, since $k_0 \geq m_{\alpha}$, by Conjecture \ref{numth}, there exist indices $1<i<j\leq k_0$ such that if we set $r'_1=r_i, r'_2=r_j
 $, then  \begin{equation} \label{inalpha1} \frac{\Big|\sum_{\ell=3}^{k_0} \Big( \frac{r'_1-r'_{\ell}}{p} \Big) \Big( \frac{r'_2-r'_{\ell}}{p} \Big)\Big|}{k_0-2}< \alpha \end{equation} 
 
 \noindent where $r'_3,...,r'_{k_0}$ are the elements of the set $\{r_1,...,r_{k_0}\} \setminus \{r'_1,r'_2\}$ (the order is irrelevant).  In the next step, we consider $I_3 :=\{r'_3,r'_4,...,r'_{k_0}\}$, and if $k_0-2=|I_3| \geq m_{\alpha}$, then  we apply Conjecture \ref{numth} again, and after possibly a permutation, we can assume that \begin{equation} \label{inalpha2} \frac{\Big| \sum_{\ell=5}^{k_0} \Big( \frac{r'_3-r'_{\ell}}{p} \Big) \Big( \frac{r'_4-r'_{\ell}}{p} \Big) \Big| }{k_0-4}< \alpha. \end{equation}

  \noindent Continuing this process, in the last step we reach  
 \begin{equation} \label{inalpha3} \frac{\Big| \sum_{\ell=k_0-b_{\alpha}+2}^{k_0} \Big( \frac{r'_{k_0-b_{\alpha}}-r'_{\ell}}{p} \Big) \Big( \frac{r'_{k_0-b_{\alpha}+1}-r'_{\ell}}{p} \Big) \Big| }{b_{\alpha}-1}< \alpha \end{equation}  
 
 \noindent if $k_0 \equiv b_{\alpha}+1$ mod 2, and \begin{equation} \label{inalpha4} \frac{\Big| \sum_{\ell=k_0-b_{\alpha}+3}^{k_0} \Big( \frac{r'_{k_0-b_{\alpha}+1}-r'_{\ell}}{p} \Big) \Big( \frac{r'_{k_0-b_{\alpha}+2}-r'_{\ell}}{p} \Big) \Big| }{b_{\alpha}-2}< \alpha \end{equation} 
 
 \noindent if $k_0 \equiv b_{\alpha} $ mod 2. Now, let $\lambda_{k_0}^{\max}$ and $\lambda_{k_0}^{\min}$ denote the maximum and minimum eigenvalues of the Gramian matrices of order $k_0$ with the largest maximum eigenvalue and the smallest minimum eigenvalues respectively. We prove the theorem using two main steps.


 \textbf{Step 1: Proving $\lambda_{k_0}^{\max} \leq 1+\frac{k_0^{\beta}}{\sqrt{p}}$}. As mentioned above, our induction base consists of verifying \eqref{deltakbet} for $k=b_{\alpha}-1, b_{\alpha}-2$ numerically. Next, we consider two cases. 
 
  Case (I): $k_0 \equiv b_{\alpha}$ mod 2. Our goal is to estimate the eigenvalues of a matrix of the form $G=G_{k_0}^{\max}=\Phi_T^* \Phi_T$, with $T=\{r_1,r_2,...,r_{k_0}\}$. After possibly a proper permutation, we can assume that $T=\{r'_1,...,r'_{k_0}\}$ is such that all of \eqref{inalpha1}, \eqref{inalpha2},...,\eqref{inalpha4}  hold.  Now, let $T_1=\{r_{k_0-b_{\alpha}+3},...,r_{k_0}\}$, $k_1 :=b_{\alpha}-2$, $G_1=\Phi_{T_1}^* \Phi_{T_1}$ ($G_1$ is obtained by considering the last $k_1=b_{\alpha}-2$ rows and columns of $G$). We start by estimating eigenvalues of this matrix. By induction hypothesis,  an upper bound for the largest eigenvalue of $G_1$ is given by $$\eta_{k_1}=1+\frac{k_1^{\beta}}{\sqrt{p}}$$ 
 
 \noindent Now, we show that $$\eta_{k_2}=1+\frac{(k_1+2)^{\beta}}{\sqrt{p}}$$ (with $k_2=k_1+2$) is an upper bound for the maximum eigenvalue of the Gramian matrix $G_2$ obtained by considering the last $k_1+2$ rows and columns of $G$.  In order to do that, we write $G_2$ in the form $$ G_2=\begin{bmatrix} 1 & b & \textbf{c} \\ b^*  & 1 &  \textbf{d} \\ \textbf{c}^* & \textbf{d}^* & G_1  \end{bmatrix} $$ with $\textbf{c}=(c_1,...,c_k)$, $\textbf{d}=(d_1,...,d_k)$, $c_i=\langle \tilde{\Phi}_{r_{k_0-k_1-1}}, \tilde{\Phi}_{r_{k_0-k_1+i}} \rangle $, and $d_i=\langle \tilde{\Phi}_{r_{k_0-k_1}}, \tilde{\Phi}_{r_{k_0-k_1+i}} \rangle $ (with $1 \leq i \leq k_1$). Note that  each of the entries $b,c_i$'s and $d_i$'s are $\pm \frac{i}{\sqrt{p}}$. (similar to what we did in \eqref{gkmax}). We also define $B_2$ via \begin{equation} \label{b2def} B_2=\begin{bmatrix} 1 & b & \textbf{c} \\ b^*  & 1 &  \textbf{d} \\ \textbf{c}^* & \textbf{d}^* & \eta_{k_1}  \end{bmatrix} \end{equation} where $\eta_{k_1}=1+\frac{D}{\sqrt{p}}$, with $D=k_1^{\beta}$, and we let  $\lambda=1+C/\sqrt{p}$ with  $C>(k_1+2)^{\beta}$. We will show that $p(\lambda) \ne 0$, where $p(x)$ is the characteristic polynomial of $B_2$.  First, we write the expression for $p(x)$ as given in \eqref{generalp}.

  $$ \begin{aligned} p(x) & = (\eta_k-x)^{k-2} \Big( (1-x))^2 (\eta_k-x)^2 -(1-x)(\eta_k-x) (\textbf{d}\textbf{d}^*+\textbf{c}\textbf{c}^*)- \\ & (\eta_k-x)^2 bb^* +2(\eta_k-x)Re(b\textbf{d}\textbf{c}^*) +\gamma \Big)  \end{aligned} $$
 
 
\noindent  Using $Re(b \textbf{d} \textbf{c}^*)=0$, $\textbf{c}^*\textbf{c}=\textbf{d}^*\textbf{d}=\frac{k_1}{p}$, and $\eta_k-\lambda \ne 0$, we observe that to show $p(\lambda) \ne 0$, it  is enough to show \begin{equation} \label{q1ineq} q(C) := C^2(D-C)^2-(-C)(D-C)(2k_1)-(D-C)^2+ \gamma>0 \end{equation} i.e., $$q(C)=C^2(C-D)^2-C(C-D)(2k_1)-(C-D)^2 +\gamma >0$$ In order to show this inequality, we show that \begin{equation} \label{forgamma} C(C-D)(2k_1)+(C-D)^2< \gamma \end{equation} where we used the fact that $\gamma \geq 0$. We call the left hand side and right hand side of the inequality above as \textit{LHS} and \textit{RHS} respectively. Next, $$LHS = (C-D) \Big( 2k_1C+(C-D) \Big) =(C-D) \Big( (2k_1+1)C-D \Big) $$ Our goal is to prove \eqref{forgamma} for any $C>(k_1+2)^{\beta}$, but we start by considering $C=(k_1+2)^{\beta}$. Then, $$C-D=(k_1+2)^{\beta}-k_1^{\beta}<2$$ where we used the fact that if we set $f(x)=a^x-(a-2)^x$, then $f(1)=2$, and $f(x)$ is increasing for $0 \leq x \leq 1$. Hence, $$LHS \leq 2 \Big( (2k_1+1)(k_1+2)^{\beta}-k_1^{\beta} \Big) < 2 \Big( (3k_1)(2k_1)^{\beta} \Big)<12 k_1^{1+\beta}  $$ where we used the fact that $2k_1+1<3k_1$, and $k_1+2<2k_1$. 

\noindent On the other hand, $$RHS=\gamma=\sum_{i=1}^{k_1} c_ic_i^* \textbf{d}_i \textbf{d}_i^*-    \sum_{i=1}^{k_1} c_i d_i^* \textbf{d}_i \textbf{c}_i^*=k_1(k_1-1)- \sum_{i=1}^{k_1} c_i d_i^* \textbf{d}_i \textbf{c}_i^*    $$  In above, $\textbf{c}_i$, $\textbf{d}_i$ are vectors $\textbf{c}$ and $\textbf{d}$ excluding their $i$th entry (so they are $(k_1-1)$-dimensional vectors). Now, we find an upper bound for $\sum_{i=1}^{k_1} c_i d_i^* \textbf{d}_i \textbf{c}_i^*$: 

$$\sum_{i=1}^{k_1} c_i d_i^* \textbf{d}_i \textbf{c}_i^* \leq \sum_{i=1}^{k_1} |\textbf{d}_i^* \textbf{c}_i| \leq k_1 \max_{i} |\textbf{d}_i^* \textbf{c}_i | $$

\noindent On the other hand, for each $1 \leq j \leq k_1$, by \eqref{inalpha4}, we have $$|\textbf{d}_j^* \textbf{c}_j| \leq  \Big| \sum_{i=k_0-k_1+1}^{k_0} \Big( \frac{r'_{k_0-k_1-1}-r'_i}{p} \Big) \Big( \frac{r'_{k_0-k_1}-r'_i}{p} \Big) \Big|+1 \leq \alpha k_1+1 $$ Thus, $$\gamma \geq k_1(k_1-1)- \alpha k_1^2 -k_1=(1-\alpha) \cdot k_1^2-2k_1$$

 \noindent  and hence,  $\gamma \geq 12 k_1^{1+\beta}$. In above, we used the fact that for $k\geq c_{\alpha}$, we have $12k^{1+\beta} \leq (1-\alpha)k^2-2k$. Therefore, $q(C)>0$ for $C=(k_1+2)^{\beta}$. Next, note that the value of $\gamma$ does not depend on $C$, and note that $C-D$ is increasing as a function of $C$. Hence, if we show that $g(C) :=C^2(C-D)-2k_1C-(C-D)  $ is an increasing function of $C$, then we can conclude that $q(C)>0$ for $C>(k_1+2)^{\beta}$. Now, we have $$\frac{d}{dC} g(C)=3C^2-2k_1-2CD-1=3C^2-2C \cdot k_1^{\beta} -2k_1-1 >0$$
 
 \noindent for $C\geq (k_1+2)^{\beta}$. To justify the last inequality, we define $h(C):=3C^2-2C \cdot k_1^{\beta}-2k_1-1$, and we verify that $h\Big( (k_1+2)^{\beta} \Big) >0$, and $\frac{d}{dC} h(C)>0$ for $C>(k_1+2)^{\beta}$ :  $$ \begin{aligned} & h \Big( (k_1+2)^{\beta} \Big) = 3 (k_1+2)^{2 \beta}-2 (k_1^2+2k_1)^{\beta}-2k_1-1 \\ & =3(k_1^2+4k_1+4)^{\beta}-2(k_1^2+2k_1)^{\beta}-2k_1-1\\ & \geq 3(k_1^2+2k_1)^{\beta}-2(k_1^2+2k_1)^{\beta}-2k_1-1=(k_1^2+2k_1)^{\beta}-2k_1-1>0, \end{aligned} $$ for $k_1 \geq 4$. We also have  $$\frac{d}{dC} (3C^2-2C k_1^{\beta}-2k_1-1)=6C-2k_1^{\beta}>0 $$ for $C>k_1^{\beta}/3$. Therefore, if we set $k_2:=k_1+2$,we have shown that $1+\frac{(k_2)^{\beta}}{\sqrt{p}}$ is an upper bound for the maximum eigenvalue of $G_2=\Phi_{T_2} \Phi_{T_2}$, with $T_2$ being the set containing the last $k_2$ entries of $T$, i.e., $T_2=T_1\cup \{r_{k-b_{\alpha}+2},r_{k-b_{\alpha}+1}\}$. After repeating a similar reasoning mentioned above, we can conclude that $1+\frac{k_3^{\beta}}{\sqrt{p}}$ is an upper bound for $G_3=\Phi_{T_3}^* \Phi_{T_3} $, with $k_3 :=k_2+2$, and $T_3$ being the set containing the last $k_3$ elements of $T$.  By continuing this process, we conclude that in the last step, $1+\frac{k_{\ell}^{\beta}}{\sqrt{p}}$ is an upper bound for the maximum eigenvalue of $G_{\ell}=G=\Phi_{T_{\ell}} \Phi_{T_{\ell}}$, where $k_{\ell}=k_0$, and $T_{\ell}=T$.


Case (II).  $k_0 \equiv b_{\alpha}+1$ mod 2. In this case,  to find an upper bound for the maximum eigenvalue of $G=G_{k_0}^{\max}=\Phi_T^* \Phi_T$ (again with $T=\{r_1,r_2,...,r_{k_0}\}$), we begin with estimating the eigenvalues of $G_1 :=\Phi_{T_1}^* \Phi_{T_1}$, with $T_1$ being the set containing the last $k_1:=b_{\alpha}-1$ elements of $T$.  By induction hypothesis, the upper bound for the largest eigenvalue of $G_1$ is given by $$1+\frac{k_1^{\beta}}{\sqrt{p}} $$ 

\noindent Next, it can be shown (similar to above) that $1+\frac{k_2^{\beta}}{\sqrt{p}}$ is an upper bound for the largest eigenvalue of $G_{2}=\Phi_{T_2}^* \Phi_{T_2}$, with $k_2:=k_1+2$, and $T_2$ being the set containing the last $k_2$ entries of $T$. Continuing this process, we conclude that the $1+\frac{k_{\ell}^{\beta}}{\sqrt{p}}$ is an upper bound for the largest eigenvalue of $G=G_{\ell}=\Phi_{T_{\ell}}^* \Phi_{T_{\ell}}$, with $k_{\ell}=k_0$ and $T_{\ell}=T$.  

 \textbf{Step 2: Proving $\lambda_{k_0}^{\min} \geq 1-\frac{k_0^{\beta}}{\sqrt{p}}$}. The proof of this step is similar to the one given for Step 1 (and consists of two cases). In each case, we start by concluding from induction hypothesis that $\eta_{1,k_1} :=1-\frac{D}{\sqrt{p}}$, with $D=k_1^{\beta}$ is a lower bound for the minimum eigenvalue of $G_1=\Phi_{T_1}^* \Phi_{T_1}$. Then, we let $\lambda=1-\frac{C}{\sqrt{p}}$ with $C>D+2$, and we show that $p(\lambda) \ne 0$, where $p(x)$ is the characteristic polynomial of the matrix $B_2'$ (obtained from the matrix $B_2$ as given in \eqref{b2def}, but $\eta_{k_1}$ replaced by $\eta_{1,k_1}$). To do so, it is enough to show \eqref{q1ineq} holds, which we already know its validity as proved in Step 1. Therefore, $\eta_{1,k_2}=1-\frac{k_2^{\beta}}{\sqrt{p}}$ (with $k_2=k_1+2$) is a lower bound for the minimum eigenvalue of $B_2'$ (and hence the minimum eigenvalue of $G_2$). Continuing this process, after $\ell$ steps, we conclude that $\eta_{1,k_{\ell}}=1-\frac{k_{\ell}^{\beta}}{\sqrt{p}}$ is a lower bound for the minimum eigenvalue of $G_{\ell}=G=\Phi_{T_{\ell}}^*\Phi_{T_{\ell}}$, with $k_{\ell}=k_0$, and $T_{\ell}=T$. 
 
 Gathering the results proved in Steps 1 and 2, we conclude that the RIP constants of $\tilde{\Phi}$ satisfy $$\delta_k =\max\{\lambda_{k}^{\max}-1,1-\lambda_{k}^{\min} \} \leq \frac{k^{\beta}}{\sqrt{p}}$$ for $b_{\alpha} \leq k \leq p^{\nu}$.  Therefore, $\delta_k<1/\sqrt{2}$ for $k \leq \min\{\frac{p^{\nu}}{2}, \frac{p^{\frac{1}{2\beta}}}{2}\}$, as desired.

\end{proof}

\section{Concluding remarks}


In CS, the performance of a measurement matrix is normally judged by the estimates for its RIP constants, since calculating the exact values of RIP constants is an NP-hard problem-- at least in a vast regime for number of measurements $m$ (vs. the sparsity level $k$). A common bound for RIP constants of a matrix, namely, $\delta_k \leq (k-1) \mu$ is derived by applying Gershgorin circle theorem on Gramian matrices. However, one should note that Gershgorin circle theorem estimates the eigenvalues of a matrix uniformly, while the RIP constants depend only on the maximum and minimum eigenvalues of the Gramian matrices. Furthermore, Gershgorin circle theorem can be applied to \textit{any} square matrix, and hence it does not use the fact that the Gramian matrices are \textit{positive semidefinite}. In this paper, we deployed the so-called Dembo bounds, which estimate the maximum and minimum eigenvalues of a positive semidefinite matrix, to improve the classical bound $\delta_k \leq (k-1) \mu$ by an additive constant for the so-called Paley tight frames.  However, we showed that this method has a great potential in general. In fact, we showed that if a particular conjecture regarding the distribution of quadratic residues holds, then we can generalize  Dembo bounds to break the square-root barrier via $k=\mathcal{O}(m^{5/7})$. We substantiated this conjecture by numerical experiments, and we also theoretically discussed it. Furthermore, we used the notion of skew-symmetric adjacency matrices and a recent (2018) result regarding a bound on the spectral radius of an oriented graph to derive a multiplicative constant improvement on the classical bound $\delta_k \leq (k-1) \mu$ for the Paley tight frames. In particular, we showed that the maximum sparsity level satisfies $2k < \frac{\pi}{2} \cdot \frac{1}{\mu \sqrt{2}}$ (opposed to $2k < \frac{1}{\mu \sqrt{2}}+1$).

\bibliographystyle{plain}
\bibliography{beyond_Gersh_01}



%
%


\end{document}